\documentclass[pra,aps,superscriptaddress,twocolumn,nopacs,longbibliography,nofootinbib]{revtex4-1}			
\usepackage[latin1]{inputenc}
\usepackage{amsthm} 
\usepackage{amsmath}
\usepackage{mathtools}
\usepackage{graphicx} 
\usepackage{subfigure} 
\usepackage{amsfonts}
\usepackage{hyperref} 
\usepackage{xcolor} 
\usepackage{verbatim}
\usepackage{hyperref}
\usepackage{soul}
\usepackage{ulem}
\usepackage{bbold}
\usepackage{enumitem}
\usepackage{bbm}
\usepackage{tikz}

\hypersetup{
	colorlinks=true,   % false: boxed links; true: colored links, false is default
	linkcolor=blue,    % color of internal links, red is default
	citecolor=blue,    % color of links to bibliography
	urlcolor=blue      % color of external links, cyan is default
}

\def\B{ {\cal B} }
\def\C{ {\cal C} }
\def\D{ {\cal D} }

\def\M{ {\cal M} }
\def\P{ {\cal P} }
\def\S{ {\cal S} }
\def\T{ {\cal T} }
\def\U{ {\cal U} }
\def\I{ {\cal I} }

\newcommand{\diag}[1]{\mathrm{diag}\left(#1\right)}
\def\>{\rangle}
\def\<{\langle}
\newcommand{\bra}[1]{\langle {#1} |}
\newcommand{\ket}[1]{| {#1} \rangle}
 
\newcommand{\ketbra}[2]{\ensuremath{\left|#1\right\rangle\!\!\left\langle#2\right|}}
\newcommand{\matrixel}[3]{\ensuremath{\left\langle #1 \vphantom{#2#3} \right| #2 \left| #3 \vphantom{#1#2} \right\rangle}}
\newcommand{\tr}[1]{\mathrm{Tr}\left( #1 \right)}
\newcommand{\trr}[2]{\mathrm{Tr}_{#1}\left( #2 \right)}
\newcommand{\norm}[1]{\left\lVert#1\right\rVert}
\newcommand{\iden}{\mathbb{1}}
\newcommand{\vect}[1]{|#1\rangle\rangle}

\renewcommand{\v}[1]{\ensuremath{\boldsymbol #1}}

\DeclareMathOperator*{\argmax}{argmax}
\newcommand*{\Scale}[2][4]{\scalebox{#1}{$#2$}}%
\theoremstyle{plain}
\newtheorem{thm}{Theorem}
\newtheorem{lem}[thm]{Lemma}

\newtheorem{prop}[thm]{Proposition}
\theoremstyle{definition}

\theoremstyle{remark}

\begin{document}
\normalem
	
\title{Coherifying quantum channels}

\author{Kamil Korzekwa}
\affiliation{Centre for Engineered Quantum Systems, School of Physics, The University of Sydney, Sydney, NSW 2006, Australia}
\author{Stanis{\l}aw Czach{\'o}rski}
\affiliation{Faculty of Physics, Astronomy and Applied Computer Science, Jagiellonian University, 30-348 Krak{\'o}w, Poland}
\author{Zbigniew Pucha{\l}a}
\affiliation{Faculty of Physics, Astronomy and Applied Computer Science, Jagiellonian University, 30-348 Krak{\'o}w, Poland}
\affiliation{Institute of Theoretical and Applied Informatics, Polish Academy of Sciences, 44-100 Gliwice, Poland}
\author{Karol {\.Z}yczkowski}
\affiliation{Faculty of Physics, Astronomy and Applied Computer Science, Jagiellonian University, 30-348 Krak{\'o}w, Poland}
\affiliation{Center for Theoretical Physics, Polish Academy of Sciences, 02-668 Warszawa, Poland}

\begin{abstract}

Is it always possible to explain random stochastic transitions between states of a finite-dimensional system as arising from the deterministic quantum evolution of the system? If not, then what is the minimal amount of randomness required by quantum theory to explain a given stochastic process? Here, we address this problem by studying possible \emph{coherifications} of a quantum channel $\Phi$, i.e., we look for channels $\Phi^\C$ that induce the same classical transitions $T$, but are ``more coherent''. To quantify the coherence of a channel $\Phi$ we measure the coherence of the corresponding Jamio{\l}kowski state $J_{\Phi}$. We show that the classical transition matrix $T$ can be coherified to reversible unitary dynamics if and only if $T$ is unistochastic. Otherwise the Jamio{\l}kowski state $J_\Phi^\C$ of the optimally coherified channel is mixed, and the dynamics must necessarily be irreversible. To assess the extent to which an optimal process $\Phi^\C$ is indeterministic we find explicit bounds on the entropy and purity of $J_\Phi^\C$, and relate the latter to the unitarity of $\Phi^\C$. We also find optimal coherifications for several classes of channels, including all one-qubit channels. Finally, we provide a non-optimal coherification procedure that works for an arbitrary channel $\Phi$ and reduces its rank (the minimal number of required Kraus operators) from $d^2$ to $d$.

\end{abstract}

\maketitle

\section{Introduction}

Random processes are ubiquitous in both classical and quantum physics. However, the nature of randomness in these two regimes differs significantly. On the one hand, classical random evolution is necessarily irreversible. On the other hand, quantum evolution may be completely deterministic (and thus reversible if no measurement is performed), but nevertheless lead to random measurement outcomes of observable $A$ by transforming a system into a coherent superposition of eigenstates of $A$. When probing the dynamics of the system one can therefore observe the same random transitions, irrespectively of whether the evolution is coherent or incoherent. The question then arises: to what extent an observed random transformation can be explained via the underlying deterministic and coherent process, and how much unavoidable classical randomness must be involved in it?

To formulate this problem more precisely, consider a \mbox{$d$-dimensional} physical system undergoing some unknown evolution. In order to characterize it, one first measures the system, finding it in some well-defined state $j$, e.g., an eigenstate of observable $A$. One then allows the system to evolve for time $\tau$ and performs the same measurement again, this time finding the system in state $i$. By repeating this procedure many times and collecting the statistics of measurement outcomes, one can reconstruct the transition matrix $T$, with entries $T_{ij}$ describing transition probabilities between states $j$ and $i$. Now, for a truly random classical process, repeating it (e.g., by letting the system evolve for~$2\tau$ instead of~$\tau$) leads to the evolution described by $T^2$. We illustrate this in Fig.~\ref{fig:classical_quantum}a for an exemplary two-dimensional system. However, in quantum physics, different transitions (paths) of $T$ can interfere with each other, so that the composition of two processes will generally not be described by a transition matrix $T^2$. In particular, the compound process can even become fully deterministic, leading to the complete disappearance of the observed randomness (see Fig.~\ref{fig:classical_quantum}b). Our question can then be rephrased as: what is the optimal \emph{coherification} of the random process described by a classical transition matrix~$T$?

A more formal motivation for our studies comes from the resource-theoretic approach to quantum information. To explain it, let us first consider a simpler and better-known problem: how coherent is a given quantum state $\rho$ of a $d$-dimensional system, and to what extent can it be transformed into a ``more coherent'' state? Note that we do not refer here to the notion of spin-coherent states, which does not depend on the choice of basis~\cite{zfg90}, but rather to a more recent concept of coherence with respect to a given basis~\cite{baumgratz2014quantifying}, distinguished for instance by the eigenbasis of the system's Hamiltonian. Any state represented by a density matrix that is diagonal in the preferred basis is incoherent in that sense, as it corresponds to a statistical mixture of classical states. On the other hand, a quantum state whose non-diagonal elements, called \emph{coherences}, do not vanish may lead to non-classical effects of quantum interference. However, a generic system-environment interaction leads to the process of \emph{decoherence}, due to which the off-diagonal entries tend to zero and the state becomes classical.

From the perspective of emerging quantum information technologies, coherence can be treated as a resource~\cite{BG15} allowing one to perform tasks impossible otherwise. It is then crucial to assess which quantum states are more valuable, i.e., have more coherence. One is thus confronted with the problem of quantifying coherence~\cite{baumgratz2014quantifying}, which effectively means ordering the set of quantum states according to their coherence properties. Several competing measures of coherence of a quantum state were recently discussed in the literature, e.g., the $l_1$ norm of the off-diagonal entries of a density matrix or the relative entropy between a state and its decohered version (for a comprehensive review see Ref.~\cite{streltsov2016quantum} and references therein). Note that the diagonal density matrices, appearing as a result of the process of decoherence, are situated at the very bottom of this ordering, as they are classical and do not carry any coherence at all.

The problem of quantifying coherence has been taken a step farther by focusing on \emph{cohering power} of quantum channels~\cite{mani2015cohering}, i.e., studying the degree to which a quantum map can create coherence in an initially incoherent state. One can analyze the maximal or the average gain of coherence, where the average is taken over a suitable set of incoherent states. Such an approach is applicable to unitary transformations~\cite{GDEP16,zanardi2017coherence} and to non-unitary operations~\cite{zanardi2017measures,ZHLF15,BKZW17}, and in this way quantum channels can be ordered depending on their power to create coherence.

Within this approach, however, quantum states and their coherence are still the central objects of interest. Here, we take an alternative path and make quantum channels themselves the main focus of our study. Employing the well known Jamio{\l}kowski--Choi isomorphism~\cite{jamiolkowski1972linear,choi1975completely}, i.e., the fact that every quantum channel $\Phi$ is isomorphic to a bipartite quantum state $J_\Phi$~\cite{zyczkowski2004duality}, we propose to apply the measures of coherence to bipartite states associated with a given channel. This way we can quantitatively investigate the problem posed at the beginning of this section: how coherent can a given random transformation be? More precisely, for a given stochastic transition matrix $T$ we look for the maximally coherent quantum channel $\Phi^\C$, which under complete decoherence collapses to $T$, so that the diagonal parts of both Jamio{\l}kowski states are equal. We first prove that a channel whose classical action is described by a transition matrix $T$ can be coherified to a reversible unitary transformation if and only if $T$ is unistochastic. We then derive general upper and lower bounds on the optimal coherification of a channel described by a non-unistochastic $T$. Finally, we construct optimally coherified maps for any one-qubit channel and certain classes of channels acting on higher dimensional systems. 

The paper is organized as follows. In Sec.~\ref{sec:setting} we set the scene by introducing necessary concepts concerning the coherence and mixedness of quantum states and channels, and formulate the optimal coherification problem. General limitations for coherifying quantum channels are then derived and analyzed in Sec.~\ref{sec:limitiations}, where we also study particular families of maps in detail. In Sec.~\ref{sec:relevance} we discuss physical interpretation of coherence and purity of a channel, and relate these quantities to unitarity~\cite{wallman2015estimating} and cohering power~\cite{mani2015cohering}. Concluding remarks are presented in the final Sec.~\ref{sec:conclusions}, while some technical results, predominantly concerning channels acting on two- and three-dimensional spaces, are relegated to Appendices~\ref{app:majo_proof}-\ref{app:polygon}.

\begin{figure}[t]
	\includegraphics[width=\columnwidth]{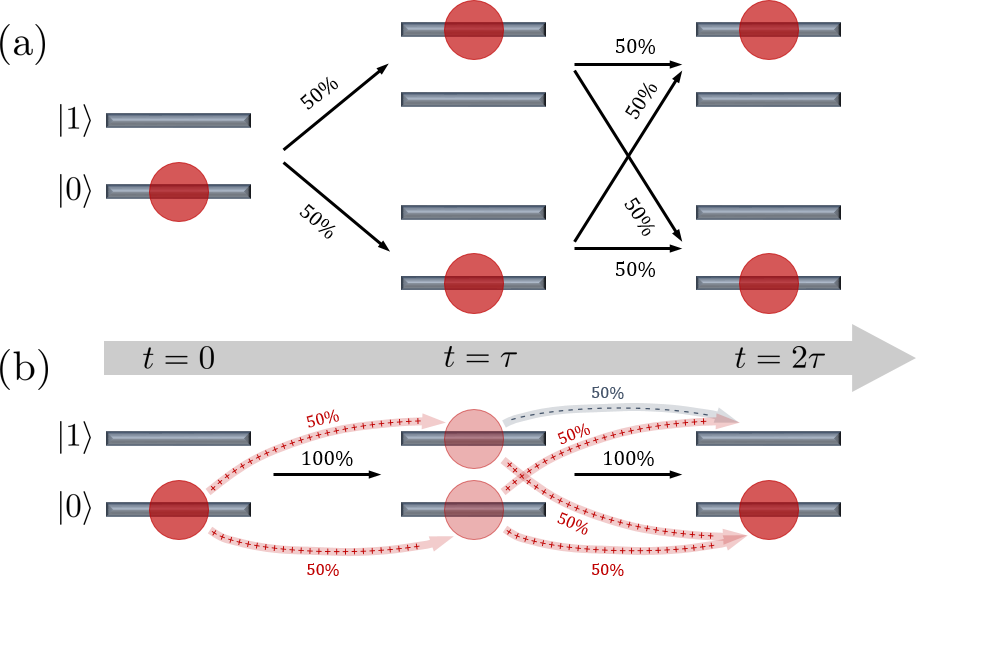}
	\caption{\label{fig:classical_quantum} \emph{Classical versus quantum randomness.} A two-dimensional system is initially prepared in a state $\ketbra{0}{0}$. (a)~The random classical evolution, running between times $0$ and $\tau$ and mapping between states $\ketbra{0}{0}$ and $\ketbra{1}{1}$, is described by the transition matrix $T$ with $T_{ij}=1/2$ for all $i,j$. The resulting state of the system at time $\tau$ is maximally mixed, \mbox{$(\ketbra{0}{0}+\ketbra{1}{1})/2$}. Further evolution between times $\tau$ and $2\tau$ is also described by $T$, leading to the total evolution being described by $T^2$ and leaving the system in the maximally mixed state. (b)~The quantum evolution between times $0$ and $\tau$ is described by a unitary operator $U$ with $U_{11}=-1/\sqrt{2}$ and $U_{ij}=1/\sqrt{2}$ otherwise (hence $U$ is a normalized $2\times 2$ Hadamard matrix). The resulting state of the system at time $\tau$ is \mbox{$\ketbra{+}{+}$}, with \mbox{$\ket{+}=(\ket{0}+\ket{1})/\sqrt{2}$}. Note that if a measurement were performed at time $\tau$, one would recover transition probabilities $T_{ij}$ as in~(a). However, if the system evolves further between times $\tau$ and $2\tau$ according to $U$, due to interference of the paths, the state of the system becomes $\ketbra{0}{0}$, and thus the total evolution is described by the identity matrix, $U^2=\iden$.}
\end{figure}

\section{Setting the scene}
\label{sec:setting}

\subsection{Coherence and mixedness of quantum states}

A state of a finite-dimensional quantum system is described by a density operator $\rho$ acting on a $d$-dimensional Hilbert space ${\cal H}_d$ that is positive, $\rho\geq 0$, and normalized by a trace condition, $\tr{\rho}=1$. The convex set of all density matrices of size $d$, denoted by $\M_d$, has $d^2-1$ dimensions and contains the $(d-1)$-dimensional simplex ${\cal P}_d$ of normalized probability vectors of length~$d$. By $\v{\lambda}(\rho)$ we will denote the probability vector with entries given by the eigenvalues of $\rho$ arranged in a non-increasing order.

A state is pure if $\rho=\rho^2$ (equivalently if \mbox{$\v{\lambda}(\rho)=[1,0,\dots,0]$}), so it can be represented by a \mbox{$1$-dimensional} projector, \mbox{$\rho=\ketbra{\psi}{\psi}$}; and mixed otherwise. Typical measures used to quantify the degree of mixedness of a given state~\cite{nielsen2010quantum,BZ17} include the von Neumann entropy,\footnote{Within this work we will use $S(\cdot)$ to denote both the von Neumann entropy of a density matrix, as well as the Shannon entropy of a probability distribution.} 
\begin{subequations}
\begin{equation}
S(\rho)=-\tr{\rho \log \rho}=-\sum_i \lambda_i(\rho)\log\lambda_i(\rho),
\end{equation}
and purity,
\begin{equation}
\gamma(\rho)=\tr{\rho^2}=\sum_{i,j}|\rho_{ij}|^2=\v{\lambda(\rho)}\cdot\v{\lambda(\rho)}.
\end{equation}
\end{subequations}
Note that, since the above measures are directly related to the eigenvalues of $\rho$, they are unitarily-invariant, and thus the mixedness of a quantum state is preserved under unitary dynamics.

On the other hand, in order to study coherence of \mbox{$\rho \in \M_d$} one first needs to specify a basis with respect to which the coherence is measured~\cite{baumgratz2014quantifying}. This basis may be distinguished by the problem under study, e.g., within quantum thermodynamics one is mainly concerned with superposition of energy eigenstates~\cite{lostaglio2015description,lostaglio2015quantum}. Here, however, we will study the problem in a general quantum information context, and thus we will simply fix an orthonormal basis $\{\ket{i}\}_{i=1}^d$. We say that a state is incoherent, or classical, when it is diagonal in the chosen basis. Classical states can be alternatively represented by a probability distribution $\v{p}=\diag{\rho}$, where $\diag{\rho}$ denotes a mapping of a density matrix $\rho$ into a probability vector $\v{p} \in \P_d$ with $p_i=\rho_{ii}$. With a slight abuse of the notation we will also write $\rho\in\P_d$ if $\rho$ is diagonal. Before we introduce measures of coherence, let us first define a \emph{completely decohering} quantum channel~$\D$, 
\begin{equation}
\label{eq:decoh_state}
\D(\rho)=\rho^\D:=\sum_i \matrixel{i}{\rho}{i}\ketbra{i}{i}.
\end{equation}
Note that under the action of $\D$ any quantum state $\rho$ undergoes complete decoherence and becomes diagonal in the preferred basis. Thus, $\rho^\D$ and the associated probability distribution $\v{p}=\diag{\rho}$ can be considered as the classical version of a general quantum state $\rho$. Notice also that $\D$ is a projector onto $\P_d$ as \mbox{$\rho^\D\in\P_d$} and \mbox{$\D(\rho^\D)=\rho^\D$} for all $\rho$.

The problem of quantifying the amount of coherence present in a state has been addressed in Ref.~\cite{baumgratz2014quantifying}, while an earlier work~\cite{aberg2006superposition} was devoted to quantifying quantum superposition. Two particular measures of coherence that we will focus on in this work are the relative entropy of coherence,
\begin{subequations}
\begin{equation}
\label{eq:c_ent}\Scale[0.92]{
\C_{\mathrm{e}}(\rho):=S(\rho||\rho^\D)=S(\rho^\D)-S(\rho)=S(\v{p})-S(\v{\lambda}(\rho)),}
\end{equation}
and the $2$-norm of coherence\footnote{We note that from the resource-theoretic perspective~\cite{baumgratz2014quantifying} $\C_{2}$ does not strictly satisfy all desirable requirements for a coherence measure. While it is true that under incoherent CPTP maps the $2$-norm of coherence is non-increasing, it can increase on average under selective measurements. However, in our study of coherence of quantum channels, the resource-theoretic constraints have no clear physical meaning, and thus $\C_2$ is a completely legitimate measure of coherence.} 
\begin{equation}
\label{eq:c_2}\Scale[0.9]{
\!\!\!\C_{2}(\rho):=\norm{\rho-\rho^\D}_{2}=\gamma(\rho)-\gamma(\rho^\D)=\v{\lambda}(\rho)\cdot\v{\lambda}(\rho)-\v{p}\cdot\v{p}.\!\!\!}
\end{equation}
\end{subequations}
It is evident that the measures of coherence are directly related to the measures of mixedness. More precisely, the relative entropy of coherence is the difference between the entropy of a classical version of a state and the quantum state itself; and the $2$-norm of coherence is the difference between the purity of a quantum state and the purity of its classical version.

Among the family of states with a fixed spectrum, i.e., belonging to a unitary orbit, the measures of mixedness are equal, but the measures of coherence vary significantly. The minimal coherence, equal to zero, is obtained for the diagonal state $U^{\dagger}\rho U$, where $U$ is the unitary matrix containing the eigenvectors of $\rho$. The maximal coherence is achieved by the contradiagonal state~\cite{lakshminarayan2014diagonal}, \mbox{$\rho^{\rm cont}=HU^{\dagger}\rho UH^{\dagger}$}, where $H$ is a Fourier matrix (or, more generally, a complex Hadamard matrix~\cite{tadej2006concise}), which is unitary and has entries with the same modulus, $|H_{ij}|^2=1/d$. Since all diagonal elements of the contradiagonal state are equal, \mbox{$\rho^{\rm cont}_{ii}=1/d$}, one gets \mbox{$\C_{\mathrm{e}}(\rho^{\rm cont})=\log d-S(\rho)$} and \mbox{$\C_{2}(\rho^{\rm cont})=\gamma(\rho)-1/d$}.

On the other hand, among the family of states with a fixed diagonal, i.e., quantum states that under the action of $\D$ decohere to the same classical state, both mixedness and coherence measures vary. However, they are maximized and minimized by the same states, which can be directly inferred from Eqs.~\eqref{eq:c_ent}-\eqref{eq:c_2}. The minimum can be obtained by acting with the decohering channel $\D$ (that leaves the diagonal unchanged) on any member of the family, leading to zero coherence. In a similar fashion we can define an optimal \emph{coherifying} transformation~$\C$ (which should not be confused with coherence measures) that maps any member of the family into a state that maximizes purity (and thus coherence),
\begin{equation}
\label{eq:coh1}
\C(\rho)=\rho^{\C} = \ketbra{\psi}{\psi},
{\rm  \ such \ that \ }
\diag{\rho^\C}= {\rm diag}(\rho).
\end{equation}
This problem has a simple solution for any mixed state~$\rho$. Identifying its diagonal elements with components of a probability vector $\v{p}$, one can write explicitly a family of optimally coherifying transformations
\begin{equation}
\label{eq:coh2}
\rho \to {\rho^\C}: \ \rho^\C_{ij} = \sqrt{p_i p_j}\; e^{i(\phi_i-\phi_j)},
\end{equation}
where the phases, $\phi_i\in [0,2 \pi)$, are arbitrary. The mixedness of such coherified states is zero and coherence achieves its maximal value, $\C_{\mathrm{e}}(\rho^\C)=S(\v{p})$ and $\C_2(\rho^\C)=1-\v{p}\cdot\v{p}$. In our work we will also refer to non-optimal coherification transformations that map a diagonal state into a state with the same diagonal but some non-zero off-diagonal terms. In Fig.~\ref{fig:coher1} we illustrate the ideas of decoherence and coherification of quantum states using low-dimensional examples.

Let us emphasize that Eq.~\eqref{eq:coh1} does not describe a realistic physical process, but rather it provides an answer to a legitimate question concerning the possible past of an irreversible quantum dynamics. Such a fictitious coherification can be treated as a kind of a formal inverse of the process of decoherence. More precisely, for a diagonal state $\rho\in\P_d$ we have $\D(\C(\rho))=\rho$; and for a pure state $\ketbra{\psi}{\psi}$ we have $\C(\D(\ketbra{\psi}{\psi}))\sim\ketbra{\psi}{\psi}$, where the equivalence is up to phase factors of the off-diagonal terms.

Finally, note that coherification can be compared to the known procedure~\cite{nielsen2010quantum,KKMB06} of \emph{purification} of a quantum state. Any mixed quantum state can be purified at the expense of increasing the dimension of the Hilbert space. More precisely, for any state $\rho \in \M_d$ its purification is given by a pure state \mbox{$\ket{\psi_{AB}}\in {\cal H}_d \otimes {\cal H}_d$} of the extended system, such that its partial trace reads \mbox{$\trr{B}{\ketbra{\psi_{AB}}{\psi_{AB}}}=\rho$}.
The Schmidt vector of $|\psi_{AB}\rangle$ coincides with the spectrum of $\rho$, e.g., if the state $\rho$ is maximally mixed, the state $|\psi_{AB}\rangle$ is maximally entangled. Both formal procedures are not unique and they allow one to find possible preimages of $\rho$ with respect to non-invertible physical operations. Namely, purification yields states of an extended system which are transformed into $\rho$ by a partial trace; and coherification of a state $\rho$ provides states of the same size which decay into $\rho^\D$ due to decoherence. 

\begin{figure}[t]
	\includegraphics[width=\columnwidth]{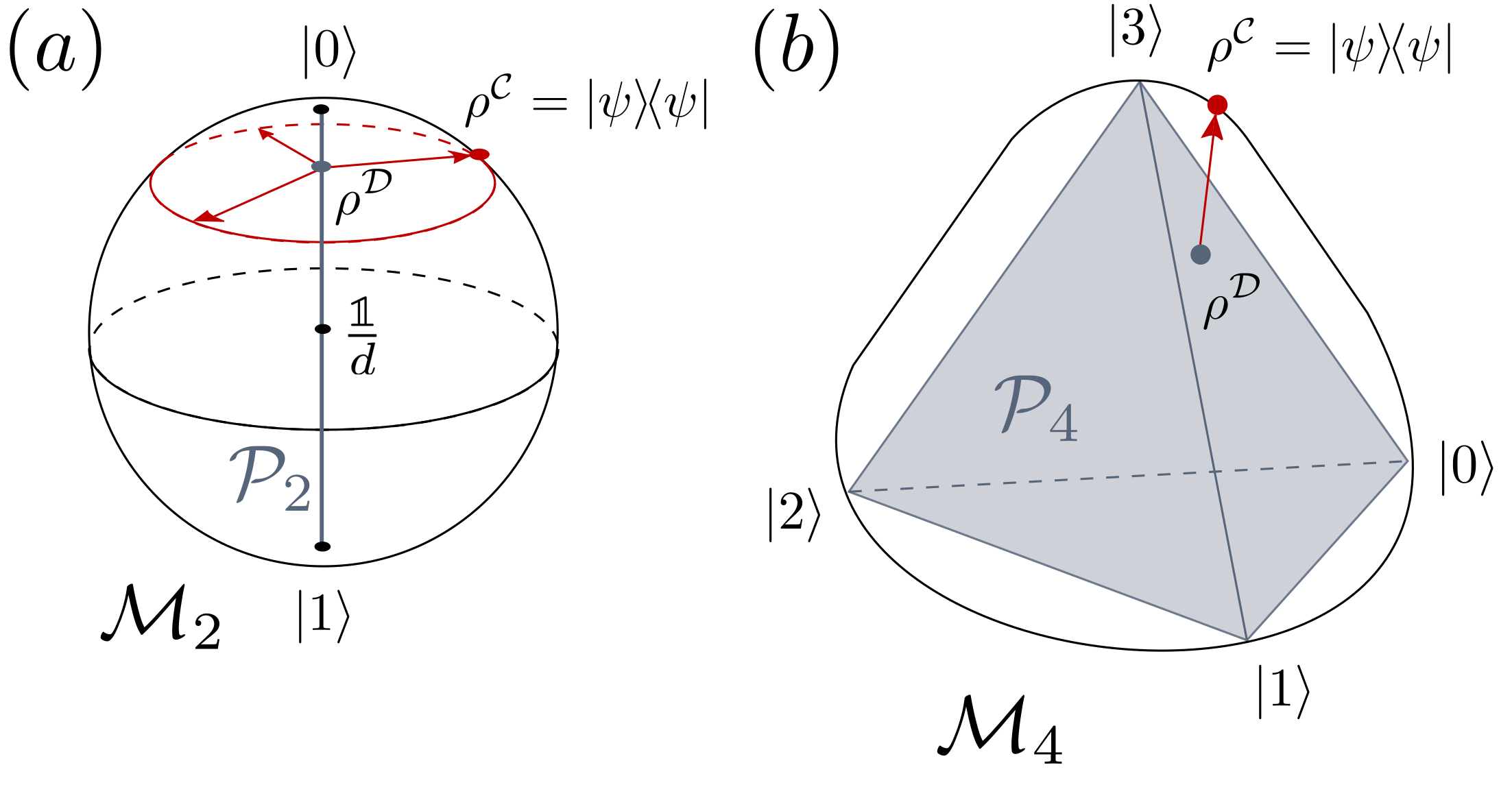}
	\caption{\label{fig:coher1} 
\emph{Decoherence and coherification of quantum states.}
(a)~For $d=2$ a pure state $\rho=\ketbra{\psi}{\psi}$ decoheres into the classical state $\rho^\D \in \mathcal{P}_2$ lying on the axis of the Bloch ball. Its coherification, $\C(\rho^\D)$ gives the entire ring of pure states $\rho^\C$ that decohere into $\rho^\D$; (b) For $d=4$ classical states from the probability tetrahedron, $\rho\in \P_4$, can be coherified into pure states $\rho^\C$ from the boundary of  $\M_4$.}
\end{figure}

\subsection{Coherence and mixedness of quantum channels}
\label{sec:coh_maps}

In this work we generalize the notion of coherence and mixedness of quantum states to quantum channels, i.e., completely positive trace preserving (CPTP) maps acting on density matrices of order $d$. We will denote by ${\cal S}_d$ the set of all quantum channels, also called stochastic maps, acting  on ${\cal M}_d$. Recall that for any $\Phi\in\S_d$ one can define the associated \emph{Jamio{\l}kowski state}~\cite{jamiolkowski1972linear}, as the image of the extended map acting on a maximally entangled state,
\begin{equation}
\label{eq:jamiol}
J_{\Phi} = \frac{1}{d}(\Phi \otimes {\I})\ketbra{\Omega}{\Omega},
\end{equation}
with \mbox{$\ket{\Omega}=\sum_i\ket{ii}$} and $\I$ denoting the identity channel. Note also that the Jamio{\l}kowski state is proportional to the dynamical matrix of Choi~\cite{choi1975completely}, so that $J_\Phi=\Phi^R/d$. Here, with a slight abuse of the notation, $\Phi$ denotes the representation of the channel as a matrix of size $d^2$, i.e., a superoperator with entries in the preferred basis given by $\Phi_{ij,kl}=\tr{\ketbra{j}{i}\Phi(\ketbra{k}{l})}$; and the reshuffling transformation, $X^R$, exchanges elements of a matrix in such a way that square blocks of size $d$ after reshaping form rows of length $d^2$, so that \mbox{$(X_{ij,kl})^R = X_{ik,jl}$} -- see Ref.~\cite{zyczkowski2004duality} for further details. Finally, for every quantum channel there exists a Kraus decomposition~\cite{nielsen2010quantum}, or the operator-sum representation, of the form
\begin{equation}
\label{eq:kraus_decomposition}
\Phi(\cdot)=\sum_i K_i(\cdot)K_i^\dagger,
\end{equation}
where $K_i$ are called Kraus operators. Due to trace preserving condition these satisfy \mbox{$\sum_iK_i^\dagger K_i=\iden$}, where $\iden$ denotes the identity matrix of size $d$.

The condition of $\Phi \in  {\cal S}_d$ is equivalent to~\cite{jamiolkowski1972linear}
\begin{subequations}
	\begin{eqnarray}
	J_\Phi&\geq& 0,\label{eq:CP}\\
	\trr{1}{J_\Phi}&=&\frac{\iden}{d}.\label{eq:TP}
	\end{eqnarray}
\end{subequations}
These conditions imply that diagonal elements (in the preferred basis) of the Jamio{\l}kowski state correspond to the entries of a \mbox{$d\times d$} (column-stochastic) transition matrix $T$,
\begin{equation}
\matrixel{ij}{J_\Phi}{ij}=\frac{1}{d}T_{ij}\label{eq:diag_stoch},
\end{equation}
with $T_{ij}\geq 0$ and \mbox{$\sum_i T_{ij}=1$}. 
The set of stochastic transition matrices of order $d$ will be denoted by ${\cal T}_d$. This set of classical maps has $d(d-1)$ dimensions and can be embedded inside the set ${\cal S}_d$ of quantum maps with $d^4-d^2$ dimensions~\cite{BZ17}. Since the diagonal of $J_\Phi$ (up to a constant $1/d$ factor) is given by the elements of a transition matrix $T$, we will write $\diag{J_\Phi}=\frac{1}{d}\vect{T}$, where $\vect{\cdot}$ denotes the (row-wise) vectorization of a matrix,
	\begin{equation}
	\label{eq:vect}
	\vect{T} = (T \otimes \iden) \ket{\Omega},
	\end{equation}
	which can be also written as
	\begin{equation}
	\vect{T} =
	\left[
	T_{11},T_{12},\dots, T_{1d},T_{21},\dots, T_{dd}\right]^\top,
	\end{equation}
	with $\top$ denoting a transpose. We also define $\langle\bra{T}$ by the Hermitian conjugate of the right hand side of Eq.~\eqref{eq:vect}.

Before we define the coherence of a quantum channel $\Phi$, let us first physically interpret the entries of the corresponding Jamio{\l}kowski state $J_\Phi$. Writing it in the matrix form in the distinguished basis we have the following block form
\begin{equation}
\label{eq:jamiol_blocks}
\renewcommand{\arraystretch}{1.3}
J_\Phi=\frac{1}{d}\left[\begin{array}{cccc}
D^1 & C^{12} & \dots & C^{1d}\\
C^{21} & D^2 & \dots & C^{2d}\\
\vdots & \vdots & \ddots & \vdots\\
C^{d1} & C^{d2} & \dots & D^d
\end{array}\right]
\end{equation}
with
\begin{equation}
\Scale[0.88]{
\renewcommand{\arraystretch}{1.3}
\!D^i=\left[\begin{array}{cccc}
T_{i1} & c^{ii}_{12} & \dots & c^{ii}_{1d}\\
c^{ii}_{21} & T_{i2} & \dots & c^{ii}_{2d}\\
\vdots & \vdots & \ddots & \vdots\\
c^{ii}_{d1} & c^{ii}_{d2} & \dots & T_{id}
\end{array}\right]\!,~
C^{ij}=\left[\begin{array}{cccc}
c^{ij}_{11} & c^{ij}_{12} & \dots & c^{ij}_{1d}\\
c^{ij}_{21} & c^{ij}_{22} & \dots & c^{ij}_{2d}\\
\vdots & \vdots & \ddots & \vdots\\
c^{ij}_{d1} & c^{ij}_{d2} & \dots & c^{ij}_{dd}
\end{array}\right]\!,\!}
\end{equation}
where $c^{ij}_{kl}=\matrixel{i}{\Phi(\ketbra{k}{l})}{j}$ 
and formally \mbox{$T_{ij}=c^{ii}_{jj}$}. We thus see that the diagonal elements of $D^i$ describe how initial populations (occupations in the preferred basis) affect the final population of a state $i$; and the off-diagonal elements of $D^i$ describe how initial coherences affect the final population of a state $i$. On the other hand, the diagonal elements of $C^{ij}$ describe how initial populations affect the final coherence between states $i$ and~$j$; and the off-diagonal elements of $C^{ij}$ describe how initial coherences affect the final coherence between states $i$ and~$j$. 

In analogy with the standard completely decohering map, Eq.~\eqref{eq:decoh_state}, we also define a decohering operation $\D$ which acts on any quantum channels $\Phi$ by bringing its corresponding Jamio{\l}kowski state into the diagonal form,
\begin{equation}
\label{eq:decoh_chan}
\Phi\to \Phi^\D: \ J_{\Phi^\D}=J_\Phi^\D.
\end{equation}
Diagonal Jamio{\l}kowski state $J_\Phi^\D\in {\cal M}_{d^2}$ represents the classical map $T \in {\cal T}_d$ acting on probability vectors of size~$d$. The action of $\Phi^\D$ on any state $\rho$ is first to completely decohere it into $\rho^\D$, and then to transform the probability vector $\v{p}=\diag{\rho}$ into $T\v{p}$, so that the final state is always classical. Therefore, again with a slight abuse of the notation, we will write \mbox{$\Phi\in\T_d$} and $\Phi\sim T$ if $J_\Phi$ is diagonal and \mbox{$\diag{J_\Phi}=\frac{1}{d}\vect{T}$}.

As every quantum channel is isomorphic to a density matrix on an extended Hilbert space, \mbox{${\cal H}_d \otimes {\cal H}_d$}, it is then natural to apply the standard measures of mixedness and coherence to the Jamio{\l}kowski state $J_\Phi$ and to characterize in this way the properties of the associated channel. More formally, for any quantum channel $\Phi$ acting on quantum states of size $d$ one defines the entropy of a channel~\cite{RFZ11},
\begin{subequations}
    \begin{equation}
	\label{eq:map_ent}
	S(\Phi)\; := \; S(J_\Phi),
	\end{equation}
and the purity of a channel,
	\begin{equation} 
        \label{eq:map_purity}
       \gamma(\Phi)\; :=\; \gamma(J_\Phi).
    \end{equation}
\end{subequations}           
These quantities allow us to introduce 
\begin{subequations}
\begin{enumerate}
	\item  \emph{entropic coherence of a channel}, 
\begin{equation} 
\label{eq:map_coh1}
\C_{\mathrm{e}}(\Phi) \; := \; \C_{\mathrm{e}}(J_\Phi)=
S\left(\dfrac{1}{d}\vect{T}\right)-S(J_\Phi),
\end{equation}	
  and
\item {\sl $2$-norm coherence of a channel},	
	\begin{align} 
		\label{eq:map_coh2}
		\C_{2}(\Phi) := \C_{2}(J_\Phi)&=\gamma(J_\Phi)-\frac{1}{d^2}\langle\langle T\vect{T}\nonumber\\
		&= \tr{J_{\Phi}^2} - \frac{1}{d^2} \tr{TT^{\dagger}}.
	\end{align}	
\end{enumerate}
\end{subequations}
Note that $\C_{2}(\Phi)$ can be decomposed into two terms,
\begin{equation}
\C_{2}(\Phi)=\C^D_{2}(\Phi)+\C^C_{2}(\Phi),
\end{equation}
with $\C^D_{2}(\Phi)$ measuring coherence coming from diagonal blocks $D^i$ and $\C^C_{2}(\Phi)$ from off-diagonal blocks $C^{ij}$, i.e.,
\begin{equation}
\C^D_{2}(\Phi)=\sum_{\mathclap{\substack{i,k,l\\k\neq l}}} |c^{ii}_{kl}|^2,\quad \C^C_{2}(\Phi)=\sum_{\mathclap{\substack{i,j,k,l\\i\neq j}}} |c^{ij}_{kl}|^2.
\end{equation}

\begin{figure}[t]
	\includegraphics[width=\columnwidth]{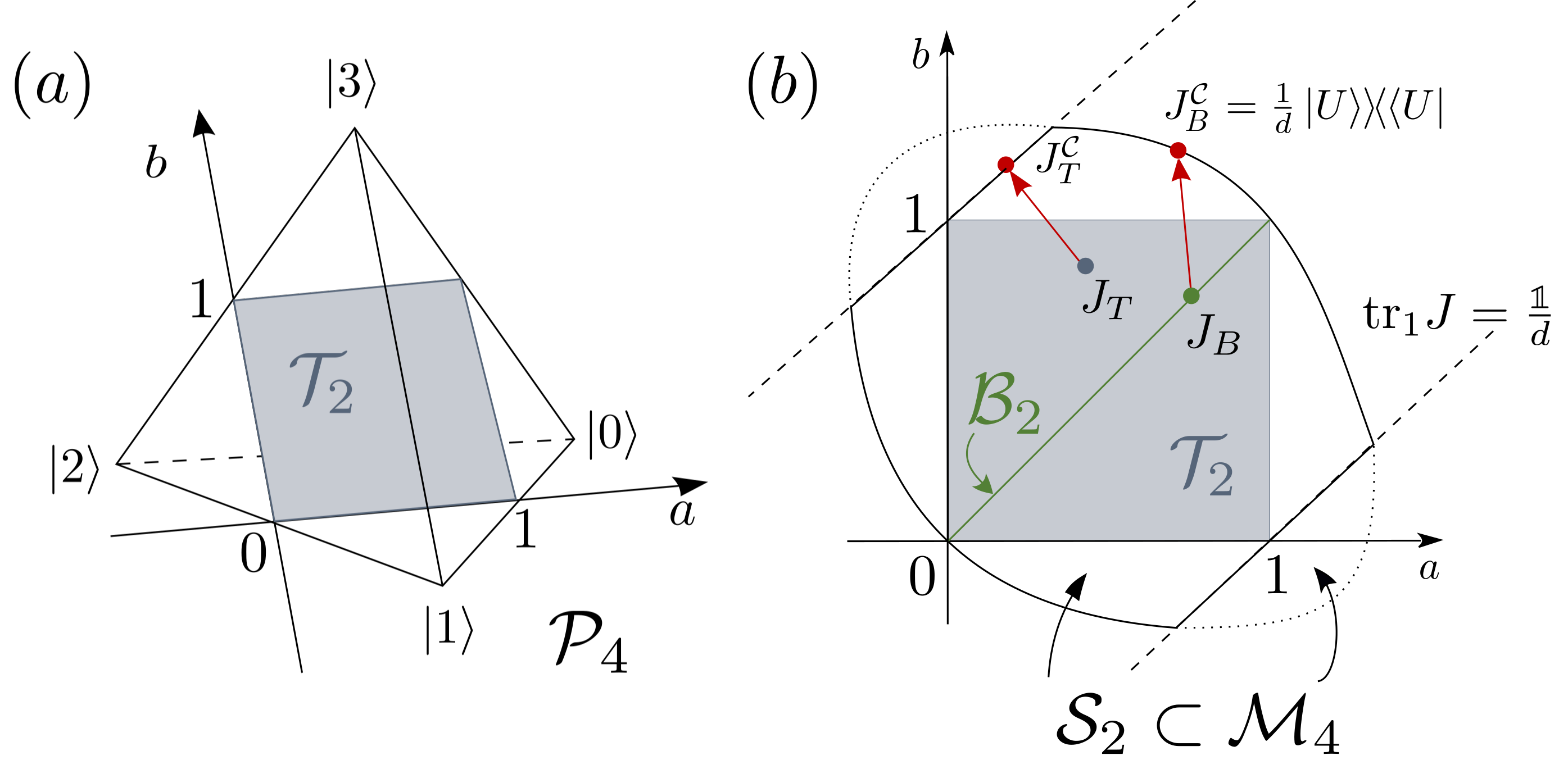}
	\caption{\label{fig:coher2} 
		\emph{Decoherence and coherification of quantum maps.}
		(a) Set ${\cal T}_2$ of stochastic matrices of size $d=2$
		embedded inside the tetrahedron $\P_4$ of
		probability vectors of length $d^2=4$;
		(b)~Optimal coherification of a quantum channel corresponding to a bistochastic matrix from $\B_2$  (and thus unistochastic) yields a unitary transformation: the Jamio{\l}kowski state $J_B$ is transformed into a pure state $J_B^\C=\frac{1}{d}|U\rangle \rangle\!\langle \langle U|$. Optimal coherification of a quantum channel corresponding to a general stochastic matrix $T \in \T_2$ yields a non-unitary channel from $\S_2$ whose Jamio{\l}kowski state $J^{\C}_T$ is mixed.}
\end{figure}

We now arrive at the central technical problem analyzed in the current work: coherifying quantum channels. Note that, given a fixed diagonal, $\diag{J_\Phi}=\frac{1}{d}\vect{T}$, of the Jamio{\l}kowski state $J_\Phi$ (equivalently: a transition matrix $T$ specifying the classical action of $\Phi$, i.e., \mbox{$\Phi^\D=T$}), one can always find the corresponding coherified pure state by simply employing the optimal coherification recipe, Eq.~\eqref{eq:coh2}. In general, however, such a pure state will not satisfy the trace preserving condition, Eq.~\eqref{eq:TP}. More precisely, this condition is equivalent to $\sum_i D^i=\iden$, and thus the choice of the off-diagonal elements of $D^i$ (describing the effect of initial coherences on final populations) is constrained beyond the standard positivity constraint. Hence, for any classical map represented by a stochastic matrix $T$ it is legitimate to ask the following question: what is the corresponding optimally coherified quantum map $\Phi^\C$ with the same classical part, such that its coherence is the largest (or its mixedness is the smallest). In Fig.~\ref{fig:coher2} we illustrate the ideas of decoherence and coherification of quantum channels using one-qubit maps as an example.

\section{Limitations of coherifying quantum channels}
\label{sec:limitiations}

In this section we will investigate the limits to which a given quantum channel $\Phi \in {\cal S}_d$, with a prescribed classical action $\Phi^{\D}=T \in {\cal T}_d$, can be coherified into an optimal channel $\Phi^{\C}$
with minimal entropy or maximal purity. To characterize potential coherification of a given channel we will simply use both coherence measures introduced above in Eqs.~\eqref{eq:map_coh1} and~\eqref{eq:map_coh2}. We thus first define optimally coherified channels according to both measures,
\begin{subequations}
	\begin{eqnarray}
	\Phi^{\C_{\mathrm{e}}}&:=&\argmax_{\Psi:~\Psi^\D=\Phi^\D}\C_{\mathrm{e}}(\Psi),\\
	\Phi^{\C_2}&:=&\argmax_{\Psi:~\Psi^\D=\Phi^\D}\C_2(\Psi),
	\end{eqnarray}
\end{subequations}
which allows us to define
\begin{subequations}
\begin{enumerate}
	\item \emph{entropic coherification}, 
		\begin{equation} 
		\label{eq:map_delta_coh1}
		\Delta\C_{\mathrm{e}}(\Phi) \; := \; \C_{\mathrm{e}}(\Phi^{\C_{\mathrm{e}}})-\C_{\mathrm{e}}(\Phi),
		\end{equation}	
	and
	\item \emph{$2$-norm coherification},	
		\begin{equation} 
		\label{eq:map_delta_coh2}
		\Delta\C_{\mathrm{2}}(\Phi) \; :=\;  \C_{2}(\Phi^{\C_{2}})-\C_{2}(\Phi).
		\end{equation}	
\end{enumerate}
\end{subequations}
We will be particularly interested in the extremal case when the coherified channel is classical, $\Phi\sim T$. Then, since \mbox{$\C_{\mathrm{e}}(T)=\C_{\mathrm{2}}(T)=0$}, we have
\begin{subequations}
	\begin{align} 
	\label{eq:map_delta_coh1_T}
	\Delta\C_{\mathrm{e}}(T) \; &= \; \C_{\mathrm{e}}(\Phi^{\C_{\mathrm{e}}})= S\left(\frac{1}{d}\vect{T}\right)- S\left(J_\Phi^{\C_{\mathrm{e}}}\right),  \\
	\label{eq:map_delta_coh2_T}
	\Delta\C_{\mathrm{2}}(T) \; &=\; \C_{2}(\Phi^{\C_{2}})=  \gamma\left(J_\Phi^{\C_2}\right) - \frac{1}{d^2}\langle\langle T\vect{T}.
	\end{align}	
\end{subequations}
Note that, although we introduced two potentially inequivalent coherification procedures, $\Phi^{\C_{\mathrm{e}}}$ and $\Phi^{\C_2}$ (with corresponding Jamio{\l}kowski states $J_\Phi^{\C_{\mathrm{e}}}$ and $J_\Phi^{\C_2}$), while deriving general bounds affecting both maximization processes, we will simply use $\Phi^\C$ (and $J_\Phi^\C$).

We now need to point out an important relation linking the classical action of a channel and its Kraus decomposition. Namely, for a channel $\Phi$ with a classical action \mbox{$\Phi^\D=T$}, the Kraus decomposition of $\Phi$, defined in Eq.~\eqref{eq:kraus_decomposition}, satisfies
\begin{equation}
\label{eq:kraus_hadamard}
\sum_i K_i\circ\bar{K}_i=T,
\end{equation}
with $\circ$ denoting the entry-wise product (also known as Hadamard or Schur product) and $\bar{K}_i$ being the complex conjugate of $K_i$.

The problem of optimal coherification of a channel naturally splits into three cases, corresponding to three families of transition matrices $T$ presented in Fig.~\ref{fig:families}. The biggest family, $\T_d$, consists of all stochastic matrices of size $d$, i.e., the most general transformations mapping the set of $d$-dimensional probability vectors into itself. These are defined by $T_{ij}\geq 0$ and $\sum_i T_{ij}=1$. The second family, $\B_d$, is given by the set of bistochastic matrices, which in addition to being stochastic satisfy $\sum_j T_{ij}=1$. This additional condition encodes the fact that bistochastic matrices map the uniform distribution, $[1/d,\dots,1/d]$, into itself. The third analyzed case corresponds to unistochastic matrices $\U_d$, which are such  bistochastic matrices $T$ whose entries can be written as $T_{ij}=|U_{ij}|^2$, for some unitary matrix $U$. Note that this condition, using Eq.~\eqref{eq:kraus_hadamard}, can alternatively be written as \mbox{$T = U \circ \bar{U}$}. While bistochastic matrices form a proper subset of stochastic matrices for all $d\geq 2$, the unistochastic matrices form a proper subset of bistochastic matrices only for $d\geq 3$, as every bistochastic matrix of order $d=2$  is unistochastic. Interestingly, the exact boundary of the set of unistochastic matrices is known only for $d\le 3$~\cite{DZ09}, while for $d\ge 3$ the set of unistochastic matrices 
is not convex~\cite{bengtsson2005birkhoff}.

\begin{figure}
	\includegraphics[width=0.7\columnwidth]{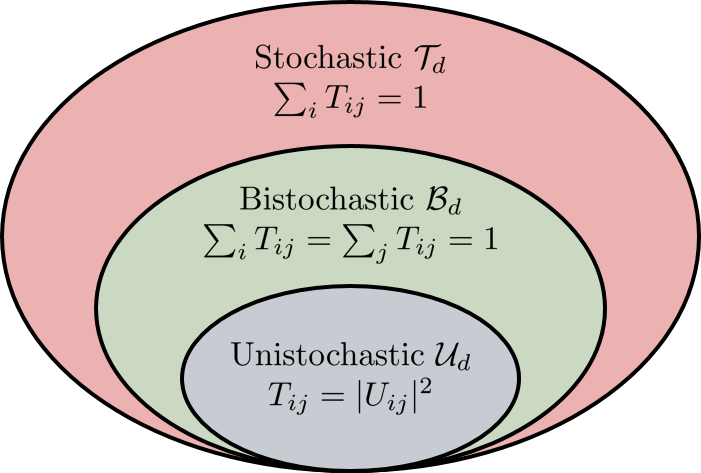}
	\caption{\label{fig:families} \emph{Families of transition matrices.} Inclusion graph of three sets of transition matrices $T$ of order $d$ describing the classical channels. For all three sets one has $T_{ij} \geq 0$. Note that $\U_2= \B_2$, while for a larger $d$ a proper inclusion relation holds, $\U_d \subset  \B_d \subset \T_d$.}
\end{figure}

\subsection{Unistochastic matrices and unitary channels}
\label{sec:unistochastic}

We start our analysis from the smallest family of unistochastic matrices. We will thus consider optimal coherification of a channel for which \mbox{$\diag{J_\Phi}=\frac{1}{d}\vect{U \circ \bar{U}}$}. We say that a given channel can be completely coherified if there exists $J^\C_\Phi$ that is pure, has the same diagonal as $J_\Phi$ and still corresponds to a valid channel. From Eq.~\eqref{eq:jamiol} it is clear that the Jamio{\l}kowski state is pure if and only if the corresponding map $\Phi$ is unitary. This simple observation can then be formalized as follows:
\begin{prop}
	\label{prop:unistoch}
	A quantum channel $\Phi$, with the corresponding Jamio{\l}kowski state $J_\Phi$, can be completely coherified to a unitary transformation if and only if its classical action is given by a unistochastic matrix.
	\end{prop}
\begin{proof}
First assume that $\Phi$ can be completely coherified. This means that there exists a pure state $J^\C_\Phi$ with $\diag{J^\C_\Phi}=\diag{J_\Phi}$, and that it corresponds to a valid channel $\Phi^\C$. However, pure Jamio{\l}kowski states correspond to unitary channels, so that [comparing Eq.~\eqref{eq:jamiol} and Eq.~\eqref{eq:vect}]
\begin{equation}
\label{eq:jamiol_uni}
J^{\C}_\Phi = \frac{1}{d}\ket{U}\rangle  \langle \bra{U}=\frac{1}{d}(U \otimes {\bar U})^R.
\end{equation}
Therefore, the diagonal of $J^\C_\Phi$ (and, by assumption, of $J_\Phi$) is given by $\frac{1}{d}\vect{U \circ \bar{U}}$, which corresponds to a unistochastic matrix. Conversely, assume that the diagonal of $J_\Phi$ is described by a unistochastic matrix $U\circ \bar{U}$. Then one can simply choose $J^\C_\Phi$ to be a pure state given in Eq.~\eqref{eq:jamiol_uni}.
\end{proof}

Notice that every non-trivial classical stochastic dynamics is irreversible. However, if it is described by a unistochastic matrix $T$, one can find a reversible (unitary) quantum channel $\Phi^\C$ whose classical action is given by $T$.  On the other hand, if the classical dynamics is not unistochastic, it cannot be completely coherified and made reversible. We will now show such a coherification procedure in action by analysing some simple examples, and in the following section we will address the limits to which a general stochastic dynamics can be made reversible.

First, consider a transition matrix given by a permutation matrix $\Pi$ of size $d$. A quantum channel $\Phi$ corresponding to a diagonal Jamio{\l}kowski state with \mbox{$\diag{J_\Phi}=\frac{1}{d}\vect{\Pi}$} is a completely decohering channel that permutes diagonal elements. The vector $\frac{1}{d}\vect{\Pi}$ of length $d^2$ has $d$ non-zero entries equal to $1/d$, so its Shannon entropy is equal to $\log d$. However, as $\Pi$ is unistochastic, it can be coherified to a unitary transformation corresponding to the Jamio{\l}kowski state \mbox{$J_\Phi^{\C} = \frac{1}{d}(\Pi \otimes \Pi)^R$} with zero entropy and $d(d-1)$ off-diagonal entries equal to $1/d$. Thus, the entropic coherification of any classical permutation matrix reads \mbox{$\Delta\C_{\mathrm{e}}(\Pi) =\log d$}, while the $2$-norm coherification is equal to $\Delta\C_{\mathrm{2}}(\Pi)=(d-1)/d$. Observe that the optimal coherification of the classical identity matrix $T={\iden}$ (corresponding to a completely decohering channel $\D$) is indeed the unitary identity quantum channel $\I$, represented by a maximally entangled state \mbox{$J_{\I}=\frac{1}{d}\ketbra{\Omega}{\Omega}$}.

Let us now move to the other extreme: the uniform van der Waerden matrix $W$ of size $d$ with entries $W_{ij}=1/d$. A quantum channel $\Phi$ corresponding to a diagonal Jamio{\l}kowski state with \mbox{$\diag{J_\Phi}=\frac{1}{d}\vect{W}$} is the completely depolarising channel, which sends any state into the maximally mixed state, \mbox{$\Phi(\rho)={\iden}/d$}. The vector \mbox{$\frac{1}{d}\vect{W}$} has $d^2$ equal entries, so that its entropy is equal to \mbox{$2\log d$}. However, as for any dimension $d$ there exists a unitary Fourier matrix $F$ with all entries of the same modulus, the uniform bistochastic matrix is unistochastic, \mbox{$W= F \circ \bar{F}$}. Thus, $\Phi$ can be completely coherified to a unitary transformation described by a pure state $J_\Phi^\C$ of zero entropy. As a result, coherification of the uniform matrix (i.e., completely depolarizing channel) is maximal, with \mbox{$\Delta\C_{\mathrm{e}}(W) = 2\log d$} and \mbox{$\Delta\C_{\mathrm{2}}(W) = (d^2-1)/d^2$}.

Finally, we consider a class of quantum channels of an arbitrary dimension that can be completely coherified: a family of Schur product channels~\cite{li1997special,levick2017quantum}, defined as
\begin{equation}
\Phi_{X} : \rho \mapsto \rho \circ X.
\end{equation}
In the above $X$ is an arbitrary correlation matrix and $\circ$, as before, denotes the entry-wise (Schur) product. The correlation matrix $X$ has ones on the diagonal to assure trace preserving condition, and positivity of $X$ guarantees complete positivity of the map $\Phi_X$. The Choi-Jamio\l{}kowski  matrix of this channel is given by
\begin{equation}
J_{\Phi_X} =\frac{1}{d}\sum_{i,j} X_{ij}\ketbra{i}{j} \otimes \ketbra{i}{j}.
\end{equation}
As $X_{ii} =1$ for all $i$, the classical action of $\Phi_X$ is given by an identity matrix. Using Proposition~\ref{prop:unistoch} we see that every Schur product channel can be completely coherified to a single common unitary channel, namely the identity.

\subsection{Stochastic matrices and majorization bounds}
\label{sec:stoch}

We now proceed to the analysis of quantum channels whose classical action is given by a general stochastic matrix $T$ that is not bistochastic (the outer shell of the set $\T_d$ presented in Fig.~\ref{fig:families}). First, we provide the majorization upper-bound on the spectrum, $\v{\lambda}(J_\Phi)$, of all Jamio{\l}kowski states with a given diagonal \mbox{$\diag{J_\Phi}=\frac{1}{d}\vect{T}$}, i.e., on the spectrum of the optimally coherified state $J^\C_\Phi$. This bound allows us to upper-bound any Schur-convex function of the spectrum $\v{\lambda}(J^\C_\Phi)$ (like the purity $\gamma(J^\C_\Phi)$) and lower-bound any Schur-concave function (like entropy $S(J^\C_\Phi)$), and thus to bound $\C_{\mathrm{e}}$ and $\C_{2}$ of the optimally coherified channel. This, in turn, is equivalent to bounding the entropic and 2-norm coherifications of a given classical channel. Next, we provide an explicit construction of a particular (non-optimal) coherified Jamio{\l}kowski state $J^{\C_0}_\Phi$, which allows us to lower bound coherence measures for the optimally coherified channel. We then illustrate the application of our results by finding optimal coherifications of qubit and qutrit channels, and interpreting their action. Finally, we make a short comment on the coherification of a particular qudit map.

\subsubsection{Upper bound for the optimal spectrum}
\label{sec:upper_bound}

Let us first recall that a probability vector $\v{p}$ is said to majorize $\v{q}$, which we denote by $\v{p}\succ\v{q}$, if and only if
\begin{equation}
\sum_{i=1}^k p_i^\downarrow\geq \sum_{i=1}^k q_i^\downarrow,
\end{equation}
for all $k\in\{1,\dots,d\}$, where $\v{p}^\downarrow$ denotes a probability vector with entries of $\v{p}$ arranged in a non-increasing order. We now state the following theorem, which we prove in Appendix~\ref{app:majo_proof} (recall that $\v{\lambda}(X)$ denotes the eigenvalues of $X$ arranged in a non-increasing order):
\begin{thm}
	\label{thm:majo}
	Given a positive semi-definite matrix $J_\Phi$ written in blocks, as in Eq.~\eqref{eq:jamiol_blocks}, we have
	\begin{equation}
	\label{eq:majo_bound}
	\frac{1}{d}\sum_{i=1}^d\v{\lambda}(D^i)\succ \v{\lambda}(J_\Phi),
	\end{equation}
	where $d(d-1)$ zeros are appended to each of the vectors $\v{\lambda}(D^i)$, so that their dimension agrees with that of $\v{\lambda}(J_\Phi)$.
\end{thm}

Next, we note that for Jamio{\l}kowski states, due to Eq.~\eqref{eq:TP}, we have $\sum_i D^i=\iden$. This results in the maximal eigenvalue of each $D^i$ to be upper-bounded by 1, as \mbox{$D^i\leq\iden$}. Consider now a $d\times d$ stochastic matrix $T$ that describes the diagonal of the Jamio{\l}kowski state $J_\Phi$. For every row of $T$ we write the sum over its columns as \mbox{$\sum_j T_{ij}=n_i+a_i$}, with $n_i$ being an integer and $a_i\in[0,1)$. We then define the following set of vectors:
\begin{equation}
\label{eq:s_vectors}
\v{s}^{(i)}(T)=\left[\phantom{\frac{i}{i}}\!\!\right.\underbrace{1,1,\dots,1}_{n_i~\mathrm{times}},a_i,0,\dots,0\left.\phantom{\frac{i}{i}}\!\!\right].
\end{equation}
Using Theorem~\ref{thm:majo} and the fact that eigenvalues of $D^i$ are upper-bounded by 1 we obtain the following majorization bound:
\begin{equation}
\label{eq:upper_bound}
\v{\mu}^\succ(T):=\frac{1}{d}\sum_{i=1}^d\v{s}^{(i)}(T)\succ \v{\lambda}(J_\Phi).
\end{equation}
Since the above majorization bound holds for all $J_\Phi$ with a fixed diagonal, in particular it also bounds the optimally coherified channel, \mbox{$\v{\mu}^\succ(T)\succ \v{\lambda}(J^\C_\Phi)$}. This can then be translated into upper-bounds on entropic coherence of $\Phi^{\C_{\mathrm{e}}}$ and 2-norm coherence of $\Phi^{\C_2}$ [so also, due to Eqs.~\eqref{eq:map_delta_coh1_T} and~\eqref{eq:map_delta_coh2_T}, on entropic and 2-norm coherifications of a classical channel $\Phi\sim T$]:
\begin{subequations}
\begin{eqnarray}
\C_{\mathrm{e}}(\Phi^{\C_{\mathrm{e}}})&\leq&S\left(\dfrac{1}{d}\vect{T}\right)-S\left(\v{\mu}^\succ(T)\right),\label{eq:C_e_bound}\\
\C_2(\Phi^{\C_2})&\leq&\v{\mu}^\succ(T)\cdot\v{\mu}^\succ(T)-\frac{1}{d^2}\tr{TT^{\dagger}}.\label{eq:C_2_bound}
\end{eqnarray}	
\end{subequations}
To illustrate the application of the bound, let us consider the following transition matrix:
\begin{equation}
\label{eq:transition_example}
T=\left[
\begin{array}{ccc}
0.7&0.2&0.6\\
0.1&0.6&0.4\\
0.2&0.2&0.0
\end{array}
\right],
\end{equation}
for which the vectors from Eq.~\eqref{eq:s_vectors} read
\begin{equation*}
\Scale[0.9]{
\v{s}^{(1)}(T)=[1,0.5,0],~\v{s}^{(2)}(T)=[1,0.1,0],~\v{s}^{(3)}(T)=[0.4,0,0].}
\end{equation*}
The bound, Eq.~\eqref{eq:upper_bound}, tells us then that \mbox{$[0.8,0.2,0,\dots,0]\succ\v{\lambda}(J^\C_\Phi)$}.

Let us also notice that the bound becomes trivial for bistochastic matrices, as in this case $\v{s}^{(i)}=[1,0,\dots,0]$ for all $i$. This, however, was to be expected, as otherwise one could differentiate between unistochastic matrices and bistochastic matrices that are not unistochastic, a problem that is known to be hard and was solved only for $d\leq 3$~\cite{bengtsson2005birkhoff}. We will come back to this problem in Sec.~\ref{sec:unital}.

\subsubsection{Lower bound for the optimal spectrum}
\label{sec:lower_bound}

We will now present a particular (non-optimal) coherification procedure $\C_0$ that can be applied to all quantum channels, irrespectively of their classical action~$T$. The coherence of channel $\Phi^{\C_0}$ coherified in such a way can thus be used as a lower-bound on the optimally coherified channel $\Phi^\C$, i.e., \mbox{$\C_{\mathrm{e}}(\Phi^{\C_{\mathrm{e}}})\geq\C_{\mathrm{e}}(\Phi^{\C_{0}})$} and \mbox{$\C_{\mathrm{2}}(\Phi^{\C_{\mathrm{2}}})\geq\C_{\mathrm{2}}(\Phi^{\C_{0}})$}.

We start by reminding that the constraint not allowing one to completely coherify a channel is the TP condition, Eq.~\eqref{eq:TP}, which means that $\sum_i D^i=\iden$. One can then choose all $D^i$ to be diagonal and try to coherify the channel only by modifying $C^{ij}$ matrices. Note that the eigenvalues of $D^i$ are then given by the entries of $T$, $\v{\lambda}(D^i)=\v{r}^{(i)}$, where the vectors $\v{r}^{(i)}$ are defined by the rows of $T$ arranged in a non-increasing order,
\begin{equation}
\v{r}^{(i)}=[T_{i1},\dots,T_{id}]^\downarrow.
\end{equation}
From Theorem~\ref{thm:majo} we thus have that
\begin{equation}
\frac{1}{d}\sum_{i=1}^d\v{r}^{(i)}\succ \v{\lambda}(J^{\C_0}_\Phi).
\end{equation}
The above majorization bound can be saturated by $J^{\C_0}_\Phi$ with the following choice of non-zero elements of $C^{ij}$. For each block $D^i$ find the maximal diagonal element, \mbox{$r^{(i)}_1$}, and set the off-diagonal elements between them (elements of $C^{ij}$) to the maximal value allowed by the CP condition, i.e., \mbox{$[r^{(i)}_1r^{(j)}_1]^{\frac{1}{2}}$}. Then, repeat the procedure for the \mbox{$n$-th} largest eigenvalues of $D^i$, $r^{(i)}_n$, with $n\in\{2,\dots,d\}$. The structure of the resulting Jamio{\l}kowski state $J^{\C_0}_\Phi$ is illustrated in Fig.~\ref{fig:lower_bound} for $d=3$. The spectrum of $J^{\C_0}_\Phi$ is given by
\begin{equation}
\label{eq:lower_bound}
\v{\mu}^\prec(T):=\frac{1}{d}\sum_{i=1}^d\v{r}^{(i)}=\v{\lambda}(J^{\C_0}_\Phi),
\end{equation}
which is also the optimal spectrum for the Jamio{\l}kowski state with a fixed diagonal, when we additionally assume no coherence in its diagonal blocks $D^i$.

\begin{figure}
	\includegraphics[width=0.6\columnwidth]{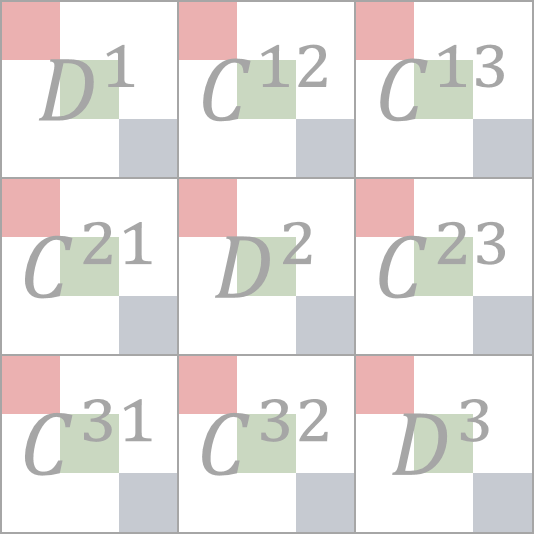}
	\caption{\label{fig:lower_bound} \emph{Structure of $J^{\C_0}_\Phi$ for $d=3$.} The non-zero entries of block matrices $D^i$ and $C^{ij}$ forming $J^{\C_0}_\Phi$ are indicated in color. Moreover, for all $i$ we have \mbox{$D^i_1\geq D^i_2\geq D^i_3$}, i.e., different colors correspond to $r^{(i)}_j$ with different $j$.}
\end{figure}

Let us now analyze the action of a channel coherified according to $\C_0$. We first note that a classical channel $\Phi=\Phi^\D$ has the following Kraus decomposition:
\begin{equation}
\Phi(\cdot)=\sum_{i,j}^d K_{ij}(\cdot)K_{ij}^\dagger,\quad K_{ij}=\sqrt{T_{ij}}\ketbra{i}{j},
\end{equation}
so that the minimal number of Kraus operators is equal to the number of non-zero entries of a stochastic matrix $T$ (in general $d^2$). On the other hand, we know that, by construction, $J^{\C_0}_\Phi$ is equal to the sum over at most $d$ projectors, so that the number of the corresponding Kraus operators will be smaller or equal to $d$. More precisely, one can obtain $i$-th Kraus operator directly from the $T$ matrix: in every row of $T$ leave only the $i$-th largest entry, replace it with its square root, and set all other entries to zero. For example, given the transition matrix from Eq.~\eqref{eq:transition_example}, we get
\begin{equation*}
K_1=\left[
\begin{array}{ccc}
\sqrt{0.7}&0&0\\
0&\sqrt{0.6}&0\\
\sqrt{0.2}&0&0
\end{array}
\right],~
K_2=\left[
\begin{array}{ccc}
0&0&\sqrt{0.6}\\
0&0&\sqrt{0.4}\\
0&\sqrt{0.2}&0
\end{array}
\right],
\end{equation*}
\begin{equation*}
K_3=\left[
\begin{array}{ccc}
	0&\sqrt{0.2}&0\\
	\sqrt{0.1}&0&0\\
	0&0&0
\end{array}
\right].
\end{equation*}
We thus see that it is always possible to coherify a channel $\Phi$, so that the number of Kraus operators (the rank of the Jamio{\l}kowski state $J_\Phi^{\C_0}$) realizing a given classical transformation $T$ (with \mbox{$\mathrm{diag}(J_\Phi^{\C_0})=\frac{1}{d}\vect{T}$}) decreases from $d^2$ to $d$. Physically, we can interpret $\C_0$ as replacing $d^2$ classical processes (transitions from state $i$ to $j$) into a classical mixture over $d$ quantum processes, where each quantum process describes a coherent superposition of $d$ classical transitions, each to a different final state. We also note that there exist stochastic matrices $T$ for which one cannot reduce the number of Kraus operators below~$d$. These are given by transitions that move all populations to a fixed $i_0$-th state, i.e., $T$ with all entries in the $i_0$-th row equal to one and all other entries being zero. The $D^i$ matrices of the corresponding Jamio{\l}kowski state are all vanishing, except for $D^{i_0}=\iden$, which cannot be coherified due to the condition \mbox{$\sum_i D^i=\iden$}.

Finally, since it is always possible to coherify a channel so that the spectrum of its Jamio{\l}kowski state is given by $\v{\mu}^\prec(T)$, one gets the following lower-bounds
\begin{subequations}
	\begin{eqnarray}
		\C_{\mathrm{e}}(\Phi^{\C_{\mathrm{e}}})&\geq&S\left(\dfrac{1}{d}\vect{T}\right)-S\left(\v{\mu}^\prec(T)\right),\\ 
		\C_{\mathrm{2}}(\Phi^{\C_{\mathrm{2}}})&\geq&\v{\mu}^\prec(T)\cdot\v{\mu}^\prec(T) - \frac{1}{d^2} \tr{TT^{\dagger}}.
	\end{eqnarray}
\end{subequations}
For comparison with the upper-bounds presented in the previous section, Eqs.~\eqref{eq:C_e_bound} and~~\eqref{eq:C_2_bound}, we note that for the exemplary transition matrix used there, Eq.~\eqref{eq:transition_example}, we have
\begin{equation*}
\v{\mu}^\prec(T)=[0.5,0.4,0.1,0,\dots,0].
\end{equation*}

\subsubsection{Qubits}
\label{sec:qubit}

Having described the general bounds on possible coherifications of quantum channels acting on arbitrary \mbox{$d$-dimensional} spaces, we now want to focus on a particular case of $d=2$. The classical action of a general qubit channel is given by a $2\times 2$ transition matrix $T$ described by two real parameters,
\begin{equation}\label{eq:def-T}
T=
\left[
\begin{array}{cc}
a&1-b\\
1-a&b
\end{array}
\right]
=:
\left[
\begin{array}{cc}
a&\tilde{b}\\
\tilde{a}&b
\end{array}
\right],
\end{equation}
with $\tilde{x}:=1-x$. We will only focus on the case when \mbox{$a\leq b$}, as the results for the case \mbox{$a>b$} are analogous. Namely, one only needs to exchange $a$ with $b$ in all expressions, and transform all matrices $X$ by replacing $X_{kl}$ with $X_{\tilde{k}\tilde{l}}$. The details of the derivations can be found in Appendix~\ref{app:qubits}. 

For $a\leq b$, Eq.~\eqref{eq:upper_bound} tells us that the spectrum of the Jamio{\l}kowski state $J_\Phi$, corresponding to any channel with classical action specified by $T$, is bounded in the following way, 
\begin{equation}
\label{eq:qubit_opt_spectrum}
\v{\mu}^\succ(T)=\frac{1}{2}[1+a+\tilde{b},b-a,0,\dots,0]\succ \v{\lambda}(J_\Phi).
\end{equation}
This bound can be in fact saturated, i.e., there exists $J^\C_\Phi$ such that $\v{\lambda}(J^\C_\Phi)=\v{\mu}^\succ(T)$. Note that, since the majorization bound is saturated, both coherification procedures (maximising entropic and 2-norm coherence) coincide, and are simply denoted by $\C$. To express the Kraus operators of the corresponding optimally coherified channel $\Phi^\C$, let us first introduce a unitary
\begin{equation}
\label{eq:qubit_unitary}
U=\frac{1}{\sqrt{a+\tilde{b}}}\left[
\begin{array}{cc}
\sqrt{a}&-\sqrt{\tilde{b}}\\
\sqrt{\tilde{b}}&\sqrt{a}
\end{array}\right],
\end{equation}
and a decaying channel \mbox{$\Psi(\cdot)=L_1(\cdot)L_1^\dagger+L_2(\cdot)L_2^\dagger$} with
\begin{equation}
\label{eq:qubit_decaying}
L_1=\left[
\begin{array}{cc}
\sqrt{a+\tilde{b}}&0\\
0&1
\end{array}\right],\quad
L_2=\left[
\begin{array}{cc}
0&0\\
\sqrt{b-a}&0
\end{array}\right].
\end{equation}
Then, \mbox{$\Phi^\C(\cdot)=K_1(\cdot)K_1^\dagger+K_2(\cdot)K_2^\dagger$} with
\begin{subequations}
\begin{align}
\label{eq:qubit_krauses1}
&K_1=L_1U^\dagger=
\left[
\begin{array}{cc}
\sqrt{a}&\sqrt{\tilde{b}}\\
-\sqrt{\frac{\tilde{b}}{a+\tilde{b}}}&\sqrt{\frac{a}{a+\tilde{b}}}
\end{array}\right],\\
\label{eq:qubit_krauses2}
&K_2=L_2U^\dagger=
\left[
\begin{array}{cc}
0&0\\
\sqrt{\tilde{a}-\frac{\tilde{b}}{a+\tilde{b}}}&\sqrt{b-\frac{a}{a+\tilde{b}}}
\end{array}\right].
\end{align}
\end{subequations}

It is straightforward to verify that the classical action of the resulting channel is given by $T$ [using Eq.~\eqref{eq:kraus_hadamard}], as well as that the spectrum of the Jamio{\l}kowski state is optimal [by checking that \mbox{$\mathrm{Tr}(K_iK_i^\dagger)=\mu_i^\succ(T)$}]. Let us also explicitly emphasize that the optimally coherified channel is given by the composition of a unitary process and a decaying channel, \mbox{$\Phi^\C(\cdot)=\Psi[U^\dagger(\cdot)U]$}. As a consequence there exists a pure state $U\ket{1}$ that is mapped by $\Phi^\C$ to a pure state $\ket{1}$, and therefore the minimum output entropy of $\Phi^\C$ is zero. We illustrate the action of $\Phi^\C$ on a Bloch sphere for a particular choice of $T$ in Fig.~\ref{fig:qubit_channel}. 

Finally, let us apply the notion of coherification to contribute to the studies on geometry of the set $\S_2$ of one-qubit stochastic maps initiated in Ref.~\cite{ruskai2002analysis}.
\begin{prop}
Coherification of any classical one-qubit stochastic map, specified by a matrix $T$ from Eq.~\eqref{eq:def-T}, yields a channel which is extremal in the set $\S_2$.
\end{prop}
\begin{proof}
In the bistochastic case $a=b$, we obtain a unitary channel, which is extremal. In other cases, without loss of generality, we may assume that \mbox{$c=b-a>0$}. The products of the Kraus operators from Eq.~\eqref{eq:qubit_decaying}, corresponding to a decaying channel, read 
\begin{equation}\label{eqn:krauss-overlaps}
\begin{array}{ll}
L_1^\dagger L_1 =\begin{bmatrix}
1-c  & 0 \\
0 & 1 
\end{bmatrix},
&  
L_1^\dagger L_2 =\begin{bmatrix}
0 & 0 \\
\sqrt{c} & 0 
\end{bmatrix},
\\
& \\
L_2^\dagger L_1 =\begin{bmatrix}
0 & \sqrt{c} \\
0 & 0 
\end{bmatrix},
& 
L_2^\dagger L_2 =\begin{bmatrix}
c & 0 \\
0 & 0 \\
\end{bmatrix}.
\end{array}
\end{equation}
Now, by direct inspection, we see that the above matrices form a linearly independent set. Thus, invoking the theorem of Choi~\cite{choi1975completely}, the channel described by two Kraus operators $L_1$ and $L_2$ is extremal. To see that $\Phi^{\mathcal{C}}$ is extremal, we note that the additional unitary matrix applied to Kraus operators will not introduce a linear dependence of the set defined in Eq.~\eqref{eqn:krauss-overlaps}.
\end{proof}

\begin{figure}
	\centering
	\begin{tikzpicture}[xscale=1,yscale=1]
	\node[] at (0,0) {
		\includegraphics[width=0.75\columnwidth]{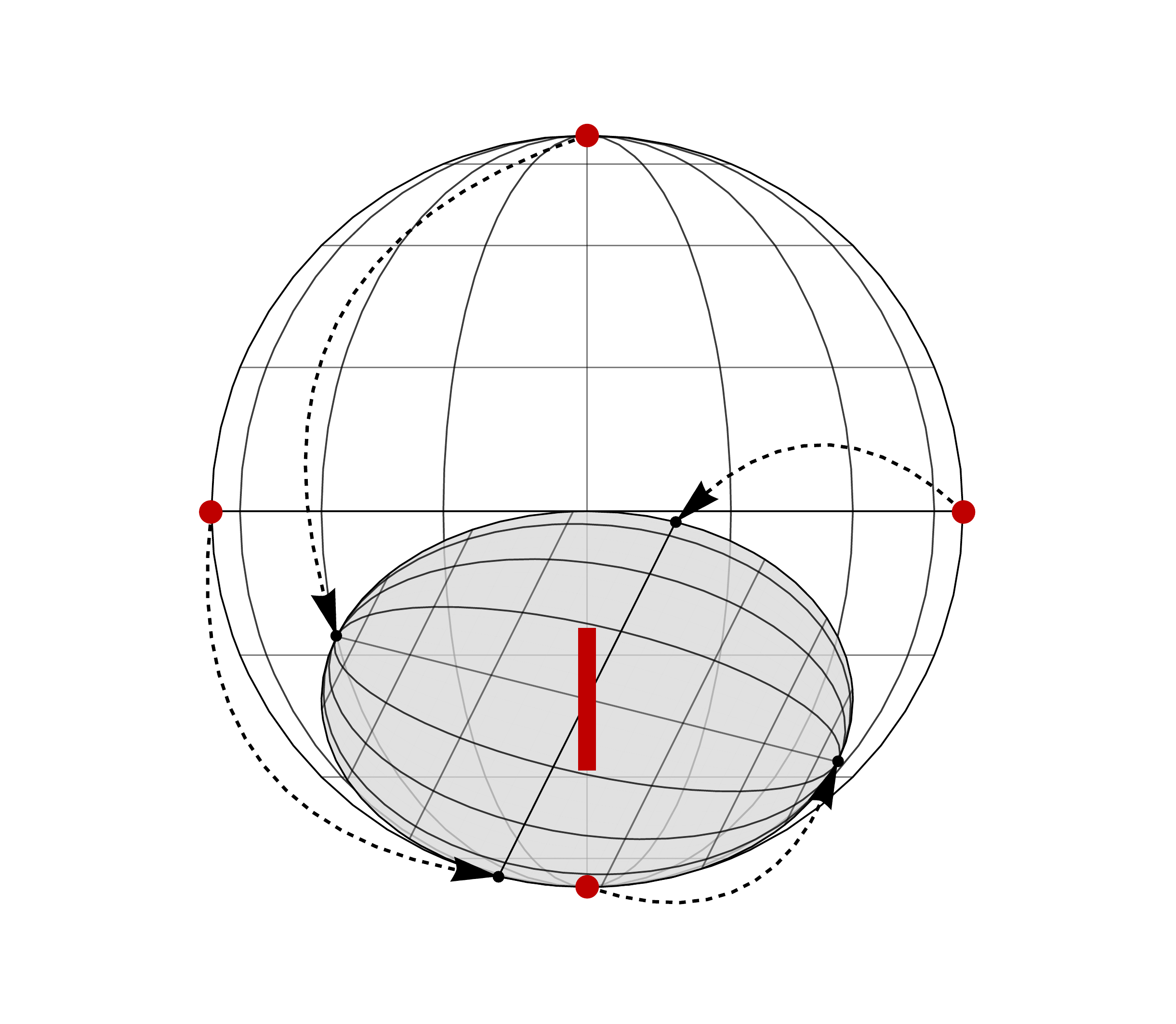}
	};
	\node[] at (0,2.4) {$\ketbra{0}{0}$};
	\node[] at (0,-2.4) {$\ketbra{1}{1}$};
	\node[] at (2.65,0) {$\ketbra{+}{+}$};
	\node[] at (-2.65,0) {$\ketbra{-}{-}$};			
	\end{tikzpicture}
	\caption{\label{fig:qubit_channel} \emph{Action of the optimally coherified qubit channel.} The image of the Bloch sphere under $\Phi^\C$ (with classical action described by $a=\frac{1}{3}$ and $b=\frac{5}{6}$) is represented by the grey ellipsoid. The thick red line represents the action of the classical channel $\Phi^\D$, while dashed lines show transformations of significant points of the sphere.}
\end{figure}
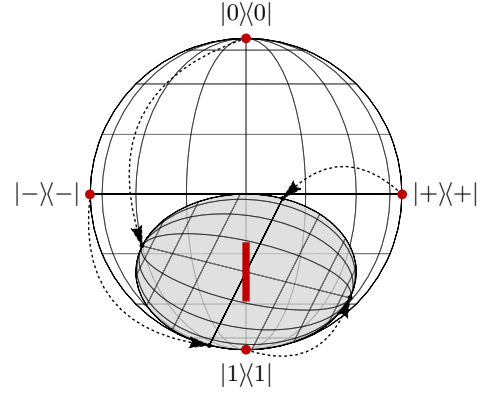

\subsubsection{Qutrits}
\label{sec:qutrit}

We will now illustrate how our results can be applied beyond the simplest qubit scenario, by using them to find optimally coherified qutrit channels $\Phi^{\C}$. Again, the coherification procedure $\C$ will optimize both considered coherence measures simultaneously. The classical action of a general qutrit channel is given by a $3\times 3$ transition matrix $T$ described by six real parameters. Here, we will consider three families of such matrices, each parametrized by three real numbers:
\begin{equation}
\label{eq:qutrit_families}
T\in\left\{
\left[\begin{array}{ccc}
0 & a & b\\
c & 0 & \tilde{b}\\
\tilde{c} & \tilde{a} & 0
\end{array}
\right],
\left[
\begin{array}{ccc}
a & b & 0\\
0 & 0 & c\\
\tilde{a} & \tilde{b} & \tilde{c}
\end{array}
\right],
\left[
\begin{array}{ccc}
a & b & c\\
\tilde{a} & \tilde{b} & \tilde{c}\\
0 & 0 & 0
\end{array}
\right]\right\},
\end{equation}\normalsize
with $\tilde{x}:=1-x$. We will refer to the above as cyclic matrices, single-row matrices and double-row matrices, respectively. For each family of $T$ we will provide the optimal spectrum of the Jamio{\l}kowski state [which yields tight bounds on $\C_{\mathrm{e}}$ and $\C_2$ via Eqs.~\eqref{eq:C_e_bound}-\eqref{eq:C_2_bound}], as well as the Kraus decomposition of the optimally coherified channel $\Phi^\C$. The details of the derivations can be found in Appendix~\ref{app:qutrits}.

For cyclic matrices $\Phi^{\C}$ is given by $\Phi^{\C_0}$, i.e., the optimal coherification procedure is given by the one we defined in Sec.~\ref{sec:lower_bound}. Introducing
\begin{equation}
\mu=\max(a,b)+\max(c,\tilde{b})+\max(\tilde{c},\tilde{b}).
\end{equation}
we thus get
\begin{equation}
\label{eq:optimal_two}
\v{\lambda}(J^\C_\Phi)=\left[\frac{\mu}{3},1-\frac{\mu}{3},0,\dots,0\right],
\end{equation}
The Kraus operators can be obtained by using the procedure described in Sec.~\ref{sec:lower_bound}, e.g., for $a\geq b$, $\tilde{b}\geq c$ and $\tilde{c}\geq\tilde{a}$ we get

\begin{equation}
K_1=\left[
\begin{array}{ccc}
0&\sqrt{a}&0\\
0&0&\sqrt{\tilde{b}}\\
\sqrt{\tilde{c}}&0&0
\end{array}
\right],~
K_2=\left[
\begin{array}{ccc}
0&0&\sqrt{b}\\
\sqrt{c}&0&0\\
0&\sqrt{\tilde{a}}&0
\end{array}
\right].
\end{equation}

For single-row transition matrices there are three separate cases depending on parameters $a$, $b$ and $c$. If $a+b\leq 1$ then the optimal spectrum is given by
\begin{eqnarray}
\v{\lambda}(J^\C_\Phi)&=&\frac{1}{3}\left[1+a+b+c,\max(\tilde{a}-b,\tilde{c}),\right.\nonumber\\
&&\phantom{\frac{1}{3}}\left.\min(\tilde{a}-b,\tilde{c}),0,\dots,0 \right],
\end{eqnarray}\normalsize
and the Kraus decomposition of the corresponding optimally coherified channel is given by
\begin{equation*}
K_1=\left[
\begin{array}{ccc}
\sqrt{a}&\sqrt{b}&0\\
0&0&\sqrt{c}\\
\sqrt{\frac{b}{a+b}}&-\sqrt{\frac{a}{a+b}}&0
\end{array}
\right],~
K_2=\left[
\begin{array}{ccc}
0&0&0\\
0&0&0\\
0&0&\sqrt{\tilde{c}}
\end{array}
\right],
\end{equation*}
\begin{equation}
K_3=\left[
\begin{array}{ccc}
0&0&0\\
0&0&0\\
\sqrt{\tilde{a}-\frac{b}{a+b}}&\sqrt{\tilde{b}-\frac{a}{a+b}}&0
\end{array}
\right].
\end{equation}
If $a+b\geq 1$ the number of non-zero elements of the optimal spectrum (so also of the Kraus operators) reduces from three to two. The optimal spectrum is again given by Eq.~\eqref{eq:optimal_two}, but this time with
\begin{equation}
\mu=1+\max(\tilde{a}+\tilde{b},\tilde{c})+c.
\end{equation}
If $\tilde{c}\geq\tilde{a}+\tilde{b}$ we thus get $\v{\lambda}(J^\C_\Phi)=[2/3,1/3,\dots,0]$, i.e., the optimal spectrum is constant for all parameters satisfying $a+b\geq 1+c$, and the Kraus operators of the optimal map are given by:
\begin{equation}
	\label{eq:kraus_single_row}
	\begin{array}{lll}
		K_1&=&\left[
		\begin{array}{ccc}
			\sqrt{\frac{\tilde{b}}{\tilde{a}+\tilde{b}}}&\sqrt{\frac{\tilde{a}}{\tilde{a}+\tilde{b}}}&0\\
			0&0&\sqrt{c}\\
			0&0&\sqrt{\tilde{c}}
	\end{array}
	\right],\\
	K_2&=&\left[
	\begin{array}{ccc}
		\sqrt{a-\frac{\tilde{b}}{\tilde{a}+\tilde{b}}}&-\sqrt{b-\frac{\tilde{a}}{\tilde{a}+\tilde{b}}}&0\\
		0&0&0\\
		\sqrt{\tilde{a}}&-\sqrt{\tilde{b}}&0
	\end{array}
	\right].
\end{array}
\end{equation}
On the other hand, if $\tilde{c}\leq\tilde{a}+\tilde{b}$ then there is a slight change in the Kraus decomposition of the optimal map. Namely, the last rows of $K_1$ and $K_2$ in Eq.~\eqref{eq:kraus_single_row} are swapped.

For double-row matrices there are again three separate cases depending on the value of \mbox{$s:=a+b+c$}. These are specified by \mbox{$s\in[0,1]$}, \mbox{$s\in(1,2)$} and \mbox{$s\in[2,3]$}, and the optimal spectra are then given by
\begin{subequations}
	\begin{eqnarray}
	\v{\lambda}(J^\C_\Phi)&=&\frac{1}{3}[1+s,1,1-s,0,\dots,0],\\
	\v{\lambda}(J^\C_\Phi)&=&\frac{1}{3}[2,1,0,\dots,0],\\
	\v{\lambda}(J^\C_\Phi)&=&\frac{1}{3}[4-s,1,s-2,0,\dots,0],
	\end{eqnarray}
\end{subequations}
respectively. Due to the lack of concise expressions, we provide Kraus decompositions of the resulting optimally coherified channels in Appendix~\ref{app:qutrits}.

\subsubsection{Qudits}
\label{sec:qudits}

Finally, we want to make a short comment about a special family of channels in the general $d$-dimensional case. Consider a \emph{completely contracting channel} $\Psi_{\sigma}$, which sends any initial state into a single point, $\Psi_{\sigma}(\rho)=\sigma$. The corresponding Jamio{\l}kowski state has a product structure and reads $J_{\Psi_\sigma}= \sigma\otimes\iden/d$~\cite{Zy08}. The output state can be coherified to a pure state $\sigma^\C=\ketbra{\psi}{\psi}$ by the standard procedure given in Eq.~\eqref{eq:coh2}. Hence the  contracting channel $\Psi_{\sigma}$, can be coherified to a channel contracting into a pure state with \mbox{$J_{\Psi_{\ket{\psi}}}= \ketbra{\psi}{\psi}\otimes\iden/d$} and zero output entropy. Note that this coherification procedure increases the entropic coherence of a channel by $S(\sigma)$. Notice also, that for a mixed state $\sigma$ such a procedure is not optimal, as can be immediately seen by recalling the result presented in Sec.~\ref{sec:unistochastic}, where we showed that $\Psi_{\iden/d}$ can be completely coherified.

\subsection{Bistochastic matrices and polygon constraints}
\label{sec:unital}

We now proceed to the analysis of quantum channels whose classical action is described by bistochastic matrices that are not unistochastic (the middle shell of the graph presented in Fig.~\ref{fig:families}). On the one hand, due to Proposition~\ref{prop:unistoch}, we know that these cannot be completely coherified. On the other, our majorization result derived in Sec.~\ref{sec:upper_bound} yields a trivial bound for bistochastic matrices. Moreover, a non-trivial constraint for all bistochastic matrices could serve as a witness of unistochasticity, and thus it is unlikely that such a concise bound can be found~\cite{bengtsson2005birkhoff}. Therefore, here we will present an approach that allows one to obtain limitations on possible coherifications of quantum channels with classical action described by a particular subset of $\B_d$.

We start by noting that due to the TP condition, Eq.~\eqref{eq:TP}, for every $k\neq l$ we have \mbox{$\sum_i D^i_{kl}=0$} (see Fig.~\ref{fig:bistochastic_constraint}a). This, via the polygon inequality, implies that
\begin{equation}
|D^i_{kl}|\leq \sum_{j\neq i} |D^j_{kl}|.
\end{equation}
Recalling that matrices $D^i$ are all positive, we have
\begin{equation}
|D^i_{kl}|\leq \sqrt{D^i_{kk}D^i_{ll}}=\sqrt{T_{ik}T_{il}}.
\end{equation}
Combining the above two equations we arrive at
\begin{equation}
\label{eq:TP_bistochastic}
|D^i_{kl}|\leq \sum_{j\neq i} \sqrt{T_{jk}T_{jl}}.
\end{equation}
We thus see that the maximum value of $|D^i_{kl}|$ allowed by CP condition is \mbox{$\sqrt{T_{ik}T_{il}}$}, whilst the TP condition restricts it via Eq.~\eqref{eq:TP_bistochastic}. Therefore, if for some $i,k,l$ we have
\begin{equation}
\label{eq:chain_relation}
\sqrt{T_{ik}T_{il}}> \sum_{j\neq i} \sqrt{T_{jk}T_{jl}},
\end{equation}
then $|D^i_{kl}|$ is constrained beyond the positivity condition and we know that the resulting Jamio{\l}kowski state cannot be pure, so that the corresponding channel cannot be completely coherified. More precisely, for every $i,k,l$ such that \mbox{$T_{ik}T_{il}\neq 0$}, we introduce
\begin{equation}
\sqrt{\alpha^i_{kl}}:=\min\left(\frac{\sum_{j\neq i} \sqrt{T_{jk}T_{jl}}}{\sqrt{T_{ik}T_{il}}},1\right),
\end{equation}
which describes the maximum fraction of the coherence (between states $k$ and $l$ of a matrix $D^i$) that could be achieved if there was no TP constraint.

\begin{figure}
	\includegraphics[width=\columnwidth]{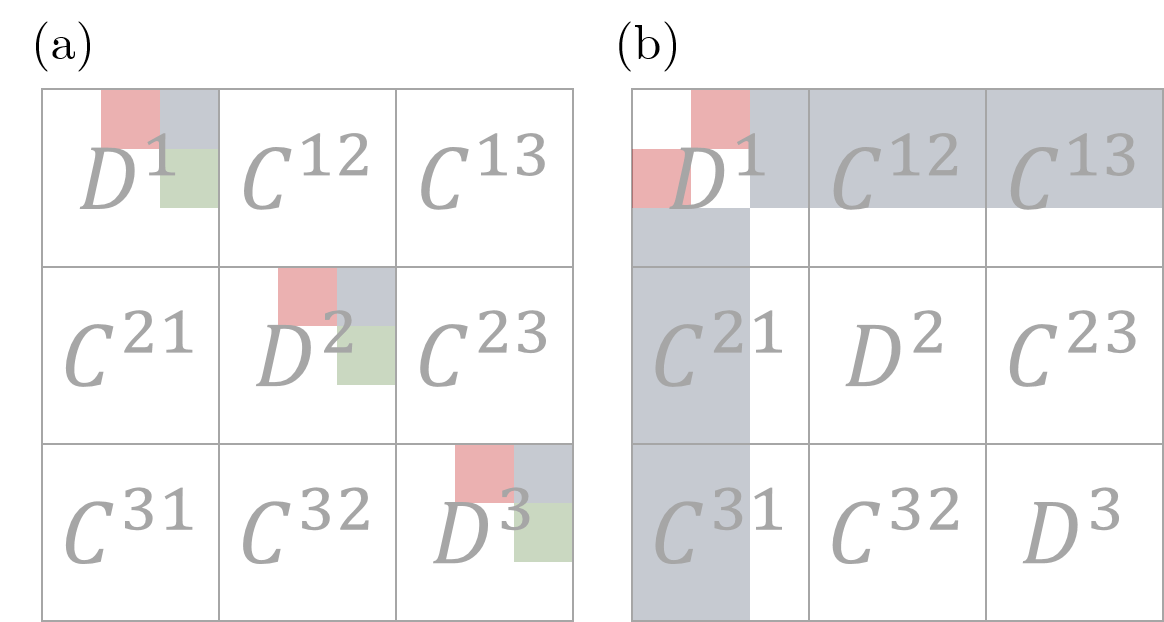}
	\caption{\label{fig:bistochastic_constraint} \emph{Constraints for channels with bistochastic classical action.} (a) Due to the TP condition, the sum of a given off-diagonal element over all $D^i$ matrices must vanish (here the summed elements are presented in the same color). (b)~If the TP condition constrains a given off-diagonal element (here: two red elements of $D^1$) beyond the positivity constraint, then also the values of all off-diagonal elements sharing a row or column index with it (here: all elements denoted in blue) are constrained beyond the positivity constraint.}
\end{figure}

Now, whenever $\alpha^i_{kl}<1$ (i.e., $|D^i_{kl}|$ must be smaller than necessary for complete coherification), other off-diagonal elements of the Jamio{\l}kowski state $J_\Phi$ also become constrained beyond the positivity condition. Before we prove this and explain how it restricts the coherification of quantum channels, let us first comment on the $\alpha^i_{kl}<1$ condition. First of all, there exist $T\in\B_d$ that are neither unistochastic, nor they satisfy this condition for any $i,k,l$. Thus, the presented bounds will, in general, work only for a subset of quantum channels with bistochastic classical action. However, as for $d=3$ a bistochastic matrix is either unistochastic or $\alpha^i_{kl}<1$ for some $i,k,l$~\cite{bengtsson2005birkhoff}, we will obtain non-trivial bounds for all qutrit channels. Further improvements would require finding a clearer separation between the sets $\U_d$ and $\B_d$.

We start by showing how $\alpha^i_{kl}<1$ can be used to constrain the purity of the optimally coherified channel. Note that Sylvester's criterion states that $J_\Phi\geq 0$ implies that all $3\times 3$ submatrices of $J_\Phi$ must have positive determinant. In particular, it means that for a part of a matrix $d\cdot J_\Phi$ containing $T_{ik}$, $T_{il}$ and any other $T_{jm}$ we have
\begin{equation}
\label{eq:Sylvesters}
\det\left[
\begin{array}{ccc}
T_{ik}&D^i_{kl}&a\\
D^{i*}_{kl}&T_{il}&b\\
a^*&b^*&T_{jm}
\end{array}
\right]\geq 0.
\end{equation}
Since we know that the maximum value of $|D^i_{kl}|^2$ is upper-bounded by \mbox{$\alpha^{i}_{kl}T_{ik}T_{il}$}, the above equation constrains all off-diagonal elements $a,b$ of $J_\Phi$ sharing a row or column with $T_{ik}$ or $T_{il}$ (see Fig.~\ref{fig:bistochastic_constraint}b). This results in the following bound on the purity of the optimally coherified Jamio{\l}kowski state $J^{\C_2}_\Phi$ (see Appendix~\ref{app:polygon} for details):
\begin{equation}
\label{eq:purity_bound}
\gamma(J^{\C_2}_\Phi)\leq 1-\Delta_1-\Delta_2,
\end{equation}
with
\begin{subequations}
\begin{eqnarray}
\label{eq:purity_deficit1}
\Delta_1&=&\frac{2}{d^2}T_{ik}T_{il}(1-\alpha^i_{kl}),\\
\label{eq:purity_deficit2}
\Delta_2&=&\frac{d-T_{ik}-T_{il}}{d^2}\left(T_{ik}+T_{il}-\beta^i_{kl}\right),
\end{eqnarray}
\end{subequations}
and
\begin{equation}
\label{eq:beta}
\beta^i_{kl}:=\sqrt{(T_{ik}-T_{il})^2+4\alpha^i_{kl} T_{ik}T_{il}}.
\end{equation}
The purity deficits, $\Delta_1$ and $\Delta_2$, add up for every $i,k,l$ for which Eq.~\eqref{eq:chain_relation} holds (however, care needs to be taken not to count twice the same terms). We illustrate this bound on the purity of the optimally coherified channel in Fig.~\ref{fig:purity_bound} for an exemplary case of a family of qutrit maps. 

\begin{figure}
	\includegraphics[width=0.75\columnwidth]{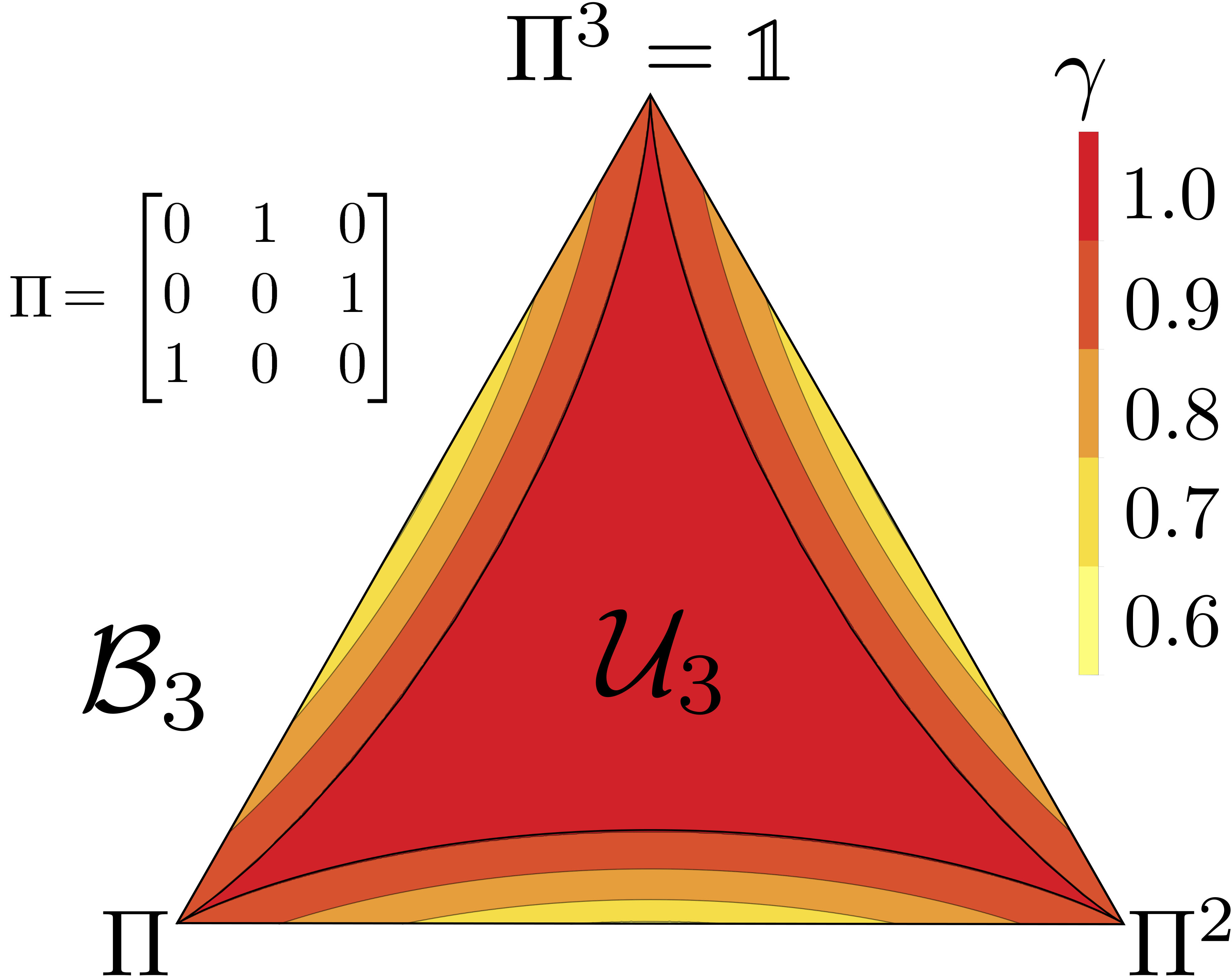}
	\caption{\label{fig:purity_bound} \emph{Purity bound.} The upper-bound on purity $\gamma$ of the family of optimally coherified qutrit channels with the classical action given by a bistochastic transition matrix \mbox{$T=\sum_{i=1}^{3}q_i\Pi^i$}, with $\sum_i q_i=1$ and $\Pi$ being a cyclic permutation matrix. Any unistochastic matrix $T \in \U_3$ can be completely coherified, so that $\gamma = 1$.
}
\end{figure}

Alternatively, one can use the fact that off-diagonal terms of $D^i$ are constrained beyond the positivity condition to bound $\v{\lambda}(D^i)$, and then use Theorem~\ref{thm:majo} to obtain a non-trivial majorization bound on the eigenvalues of the Jamio{\l}kowski state. In Appendix~\ref{app:polygon} we show for example that
\begin{equation}
\label{eq:majo_bistoch}
[1-\mu^i,\mu^i,0,\dots,0]\succ \v{\lambda}(D^i),
\end{equation}
with
\begin{equation}
\mu^i=\frac{1}{2}(T_{ik_i}+T_{il_i}-\beta^i_{k_il_i}),
\end{equation}
where $k_i$ and $l_i$ are indices for which $\alpha^i_{k_il_i}$ is minimized (so that we obtain non-trivial majorization bounds on the spectra of $D^i$, whenever $\alpha^i_{kl}<1$ for some $k$ and $l$). This, in turn, means that
\begin{equation}
\label{eq:majo_bistoch2}
\left[1-\frac{1}{d}\sum_i \mu^i,\frac{1}{d}\sum_i \mu^i,0,\dots,0\right]\succ \v{\lambda}(J^{\C}_\Phi),
\end{equation}
which can be used to obtain bounds on $\C_{\mathrm{e}}$ and $\C_2$ via Eqs.~\eqref{eq:C_e_bound}-\eqref{eq:C_2_bound}. As an example consider a quantum channel with classical action given by
\begin{equation}
\label{eq:qutrit_bistoch_ex}
T=\frac{1}{2}\left[
\begin{array}{ccc}
0&1&1\\
1&0&1\\
1&1&0
\end{array}\right],
\end{equation}
corresponding to the middle-point between $\Pi$ and $\Pi^2$ in Fig.~\ref{fig:purity_bound}. One then gets
\begin{equation}
\alpha^1_{23}=\alpha^2_{13}=\alpha^3_{12}=0,
\end{equation}
resulting in $\mu^i=1/2$, so that the spectrum of the optimally coherified Jamio{\l}kowski state is majorized by \mbox{$[1/2,1/2,0,\dots,0]$}.

Finally, let us note that, in particular cases, the tools introduced in Sec.~\ref{sec:stoch} can also be used to find limitations on coherifying channels with $T\in\B_d$. As an example consider the family of qutrit channels described by cyclic matrices [first entry of Eq.~\eqref{eq:qutrit_families}] with \mbox{$a=\tilde{b}=\tilde{c}$}. The matrix $T$ is then bistochastic and the spectrum of the optimally coherified Jamio{\l}kowski state is given by \mbox{$[a,1-a,0,\dots,0]$} (for $a\geq 1/2$) or \mbox{$[\tilde{a},1-\tilde{a},0,\dots,0]$} (for $a< 1/2$). This shows that the majorization bound, Eq.~\eqref{eq:majo_bistoch2}, applied to the channel with $T$ specified by Eq.~\eqref{eq:qutrit_bistoch_ex} is tight.

\section{Physical interpretation}
\label{sec:relevance}

We started this paper asking about the extent to which a given random transformation can be explained via the underlying deterministic and coherent process. Now, being equipped with formal bounds limiting possible coherifications of quantum channels, we will try to address this initial question. We will also provide interpretation of the purity of a channel by relating it to the notions of \emph{unitarity} and average output purity. Finally, we will comment on the links and differences between our approach to the study of coherence of quantum channels, and the ones existing in the literature.

Let us start by recalling that the evolution of a pure quantum state $\ket{\psi}$ under the action of a channel $\Phi$, described by
\begin{equation}
\Phi(\ketbra{\psi}{\psi})=\sum_i K_i\ketbra{\psi}{\psi} K_i^\dagger,
\end{equation}
can be interpreted as an incoherent (probabilistic) mixture of different pure state transformations,
\begin{equation}
\ket{\psi}\xrightarrow{\Phi}\frac{1}{\sqrt{q_i}}K_i\ket{\psi}\mathrm{~with~probability~}q_i,
\end{equation}
where
\begin{equation}
q_i=\tr{K_i\ketbra{\psi}{\psi} K_i^\dagger},
\end{equation}
and we refer to the canonical Kraus form in which all Kraus operators $K_i$ are mutually orthogonal, as they are obtained by reshaping eigenvectors of the Choi-Jamio{\l}kowski matrix. Each independent path, described by $K_i$ and being chosen with probability $q_i$, describes a coherent evolution, as it preserves the ability of a state to interfere (it maps a pure state to a pure state). Thus, the probability distribution over different paths, $\v{q}$, can be seen as describing the incoherent randomness associated with $\Phi$. Note, however, that the probability of evolving along a given path depends on the initial state of the system $\ket{\psi}$. In order to achieve a state-independent statement characterizing a given quantum channel $\Phi$, one can then focus on the average probability of choosing a given path. Introducing the average over (Haar distributed) pure states,
\begin{equation}
\langle\cdot\rangle_\psi=\int d\psi(\cdot),
\end{equation}
we see that the probability describing which path is chosen is on average\footnote{We note that the average here could actually be taken over all states, pure and mixed. However, in order to be consistent with the averaging process used in the definition of unitarity, we restrict the average to pure states only.}
\begin{equation}
\Scale[0.92]{
\langle q_i\rangle_\psi=\tr{K_i\langle\ketbra{\psi}{\psi}\rangle_\psi K_i^\dagger}=\tr{K_iK_i^\dagger}=\lambda_i(J_\Phi),}
\end{equation}
where $\v{\lambda}(J_\Phi)$ denotes, as usual, the eigenvalues of the Jamio{\l}kowski state $J_\Phi$ corresponding to $\Phi$.

We thus see that the incoherent randomness of the evolution coming from the random choice of different paths is (on average) described by the spectrum of $J_\Phi$. The extent to which a quantum channel with a given classical action $T$ can be coherified tells us how coherent the underlying evolution, leading to transitions described by $T$, can be. On one extreme, we have unistochastic transitions that can be completely coherified and, therefore, explained by a single deterministic path (unitary dynamics). On the other hand, the majorization upper-bounds on $\v{\lambda(J_\Phi)}$ that we derived in Sec.~\ref{sec:limitiations}, yield lower-bounds on the randomness of path distribution of the underlying process necessary to induce classical transformation $T$. Moreover, our majorization lower-bound provides a particular coherent explanation of every classical process~$T$ (decreasing the randomness of path distribution) and, in particular, shows that all transitions $T$ can be explained with at most $d$ paths.

Let us now focus on a particular measure describing the randomness of the path distribution, namely on the purity of a channel $\gamma(\Phi)$. One could be tempted to think that the bigger $\gamma(\Phi)$ is, the purer the average output purity,
\begin{equation}
\langle\gamma_\Phi\rangle:=\langle\gamma(\Phi(\ketbra{\psi}{\psi}))\rangle_\psi,
\end{equation}
will be. Although the two notions are related, as we will shortly see, they are not in direct correspondence. As an illustrative example consider two quantum channels,
\begin{subequations}
\begin{eqnarray}
\!\!\!\!\!\!\!\!\!\Phi_1(\cdot)&=&U(\cdot)U^\dagger\mathrm{~~with~~}U=\frac{1}{\sqrt{d}}\sum_{k,l}e^{\frac{2\pi ikl}{d}}\ketbra{k}{l},\\
\!\!\!\!\!\!\!\!\!\Phi_2(\cdot)&=&\ketbra{\psi_+}{\psi_+}\mathrm{~~with~~}\ket{\psi_+}=\frac{1}{\sqrt{d}}\sum_i\ket{i},
\end{eqnarray}
\end{subequations}
with classical action given by the van der Waerden matrix $T=W$ with flat entries \mbox{$W_{ij}=1/d$}, which maps every probability distribution to a uniform one. Both channels have the same average output purity equal to one; but for a reversible $\Phi_1$ we have \mbox{$\gamma(\Phi_1)=1$}, while for an irreversible $\Phi_2$ we get \mbox{$\gamma(\Phi_1)=1/d$}. This suggests that the purity of a channel is somehow related to reversibility of the process, which leads us to the concept of unitarity. It was originally introduced in Ref.~\cite{wallman2015estimating} to measure the departure of a channel from the unitary dynamics, and for trace-preserving channels is defined by:
\begin{equation}
u(\Phi):=\frac{d}{d-1}\left[\langle\gamma_\Phi\rangle-\gamma(\Phi(\iden/d))\right].
\end{equation}
We note in the passing that one can also relate unitarity to the variance of the random variable \mbox{$X=\Phi(\ketbra{\psi}{\psi})$}:
\begin{equation}
u(\Phi)=\tr{\langle X^2\rangle_\psi-\langle X\rangle_\psi^2}=\tr{\mathrm{Var}_\psi(X)}.
\end{equation}
For our exemplary channels we see that \mbox{$u(\Phi_1)=1$} and \mbox{$u(\Phi_2)=0$}, in accordance with the purity of the channel and capturing the fact that a completely irreversible process is as far as possible from a unitary transformation.

We will now formally relate $\gamma(\Phi)$, $\langle\gamma_\Phi\rangle$ and $u(\Phi)$. The authors of Ref.~\cite{wallman2015estimating} showed that
\begin{equation}
u(\Phi)=\frac{d}{d^2-1}\left[d\gamma(\Phi)-\gamma(\Phi(\iden/d))\right],
\end{equation}
which, using the definition of unitarity, directly leads to
the general expression for the average output purity 
derived by Cappellini~\cite{cappellini2007unpublished},
\begin{equation}
\langle\gamma_\Phi\rangle=\frac{d}{d+1}\left[\gamma(\Phi)+\gamma(\Phi(\iden/d))\right].
\end{equation}
We thus see that both average output purity and unitarity are proportional to the purity of a channel corrected by a term describing the purity of the transformed maximally mixed state. Moreover, for large $d$, $u(\Phi)$ actually approaches $\gamma(\Phi)$. Now, by noting that the minimal value of purity for a $d$-dimensional system is $1/d$, we obtain the following inequalities:\footnote{Note that by restricting to unital channels these inequalities actually become equalities.}
\begin{subequations}
\begin{align}
u(\Phi)\leq\frac{d^2}{d^2-1}&\left[\gamma(\Phi)-\frac{1}{d^2}\right],\label{eq:unitarity_bound}\\
\langle\gamma_\Phi\rangle\geq\frac{d}{d+1}&\left[\gamma(\Phi)+\frac{1}{d}\right]\label{eq:output_purity_bound}.
\end{align}
\end{subequations}
Using our majorization and purity upper-bounds we can thus upper-bound the optimal unitarity of a channel with a given classical action $T$. On the other hand, using the majorization lower-bound, we can lower-bound the average output purity of such an optimal channel. The above bounds can actually be tightened by noting that
\begin{equation}
\label{eq:max_mixed_purity}
\gamma(\Phi(\iden/d))\geq \gamma(\D(\Phi(\iden/d)))=\sum_i\left(\frac{1}{d}\sum_j T_{ij}\right)^2.
\end{equation}
Let us also mention that the optimally coherified qubit channel $\Phi^\C$, with Kraus operators specified by Eqs.~\eqref{eq:qubit_krauses1}~and~\eqref{eq:qubit_krauses2}, not only maximizes purity, but also minimizes the output purity for the maximally mixed state [as it saturates the bound in Eq.~\eqref{eq:max_mixed_purity}]. Therefore, it maximizes unitarity among all qubit channels with the same classical action.

Finally, we would like to relate the work presented here to studies on \emph{cohering power} $P$~\cite{mani2015cohering} and \emph{coherence generating power} $\tilde{P}$~\cite{zanardi2017coherence,zanardi2017measures} of quantum channels. These notions were introduced within the framework of resource theory of coherence~\cite{baumgratz2014quantifying} and measure the ability of a channel to transform initially incoherent state to a coherent one. More formally they are defined by 
\begin{subequations}
	\begin{eqnarray}
		P(\Phi)&=&\max_{\rho\in \P}\C_x(\Phi(\rho)),\\
		\tilde{P}(\Phi)&=&\langle \C_2(\Phi(\rho))\rangle_{\P},		
	\end{eqnarray}
\end{subequations}
where $\P$ denotes the set of incoherent states ($\rho\in\P$ if and only if $\matrixel{i}{\rho}{j}=0$ for $i\neq j$), $\langle\cdot\rangle_{\P}$ denotes the average over all incoherent states, and $\C_x$ is any measure of coherence for states, e.g., the relative entropy of coherence~$\C_{\mathrm{e}}$. Since both definitions involve only the action of $\Phi$ on incoherent states, we see that the only relevant parameters (defining the values of $P$ and $\tilde{P}$) are given by the diagonals of $D^i$ and $C^{ij}$, which are not constrained by the TP condition. Hence, given a fixed classical action $T$ of $\Phi$ (so fixed diagonals of $D^i$), we can choose the diagonals of $C^{ij}$ to be maximal possible (constrained only by complete positivity condition, $J_\Phi\geq 0$):
\begin{equation}
C^{ij}_{kk}=\sqrt{T_{ik}T_{jk}},
\end{equation}
and set all other matrix elements of $J_\Phi$ to zero. This way we will obtain a channel that maximizes both $P$ and $\tilde{P}$ among all channels with a fixed classical action $T$. The action of such an optimal map is given by
\begin{equation}
\Phi(\cdot)=\Phi(\D(\cdot))=\sum_j \matrixel{j}{\cdot}{j} \ketbra{\psi_j}{\psi_j},
\end{equation}
with
\begin{equation}
\ket{\psi_j}=\sum_{i} \sqrt{T_{ij}}\ket{i}.
\end{equation}
We note that the above channels that maximize $P$ and $\tilde{P}$ for a fixed $T$ do not coincide with optimally coherified channels studied in this work. The reason for this is that the latter optimization depends on all coherence terms, whereas the former one only on the ones lying on the diagonal of $C^{ij}$. This emphasizes the main difference between our channel-oriented approach (when one focuses on the properties of the channel itself, specifically how close it is to a unitary evolution) and the state-oriented approach used in the studies of cohering power and coherence generating power (when one focuses on the properties of the output states for a restricted set of input states).

\section{Conclusions}
\label{sec:conclusions}

Any classical state of size $d$, represented by a diagonal density matrix $\rho=\rho^\D$, can be coherified to a pure state $\rho^\C$ with maximal coherence, which is transformed back into $\rho$ by decoherence (see Fig.~\ref{fig:coher1}). In a similar way, one can try to coherify a quantum operation $\Phi$ represented by the corresponding Jamio{\l}kowski state $J_{\Phi}$. However, due to the trace preserving condition, the problem of coherifying a quantum channel has a much richer structure.

In this work we posed a general question: how to coherify a given classical map, represented by a stochastic transition matrix $T$, in an optimal way? Physically, this can be understood as looking for the most coherent (deterministic) underlying quantum evolution that can explain the observed random transformation $T$. Mathematically, among all quantum channels that decohere to $T$ we looked for the one whose Jamio{\l}kowski state has maximal coherence (as measured by entropic and 2-norm coherence). We demonstrated that the complete coherification to a (reversible) unitary channel is possible if and only if $T$ is unistochastic, as schematically visualized in Fig.~\ref{fig:coher2}. To capture the limitations of possible coherifications of non-unistochastic maps we derived explicit bounds for the purity and entropy of the optimally coherified channel. Furthermore, we provided an explicit coherification procedure that allows one to lower-bound the coherence of the optimal channel, and solved the optimal coherification problem for several classes of channels, including all one-qubit channels.

Studying possible coherifications of quantum channels can also shed some light on the structure and geometry of the set ${\cal S}_d$ of quantum operations~\cite{ruskai2002analysis}. For $d=2$ the set of pure quantum states (the Bloch sphere) can be obtained by coherifying the set ${\cal P}_2$ of one-bit classical states. Analogously, the square ${\cal T}_2$ of classical stochastic matrices forms a skeleton of the larger set ${\cal S}_2$ of one-qubit quantum operations. Any unistochastic matrix \mbox{$B\in\B_2=\U_2$} can be coherified into a quantum unitary transformation, corresponding to a pure Jamio{\l}kowski state $J^{\C}_B$. Furthermore, we have demonstrated that any classical transition matrix $T\in {\cal T}_2$ can be coherified to an optimal quantum channel, corresponding to a mixed state $J^{\C}_T$, that is an extremal point of ${\cal S}_2$. One would then like to check under what conditions a similar statement holds for higher dimensions, i.e., when the optimally coherified channels are extremal and have vanishing minimum output entropy.

Besides this problem concerning the geometry of $\S_d$, there are also other open questions that we would like to conclude this paper with. One could ask whether the optimally coherified channels are unique up to a unitary equivalence, i.e., can one find two channels whose Jamio{\l}kowski states are not connected via unitary, and which maximize a given coherence measure among all Jamio{\l}kowski states with a fixed diagonal? Furthermore, the expressions that lower- and upper-bound possible coherifications can definitely be improved, especially for bistochastic matrices. In this special case, exploring the boundary between unistochastic and bistochastic maps could be beneficial. Moreover, one might pursue a statistical approach and ask a question concerning a possible degree of coherification of a random stochastic matrix, or a generic quantum channel~\cite{BCSZ09}. Last but not least, it would be very interesting to establish a closer connection between coherification approach to quantum channels, pursued in this work and based on the coherence of the corresponding Jamio{\l}kowski states, with the earlier notion of the coherence power, related to the increase of coherence of selected quantum states by the action of a channel~\cite{mani2015cohering,GDEP16,zanardi2017measures}.

\textit{Note:} Shortly after our work appeared on arXiv, another preprint studying the coherence of quantum channels was posted there~\cite{datta2017coherence}.

\section*{Acknowledgements}

We would like to thank Valerio Cappellini for sharing with us his unpublished notes
and Antony Milne for his invaluable linguistic advices. We acknowledge financial support from the ARC via the Centre of Excellence in Engineered Quantum Systems, project number CE110001013 (K.K.) and Polish National Science Centre under the project numbers 2016/22/E/ST6/00062 (Z.P.) and DEC-2015/18/A/ST2/00274 (K.{\.Z}.).

\appendix

\section{Proof of Theorem~\ref{thm:majo}}
\label{app:majo_proof}

We will make use of the following known results (see Lemma~3.4 of Ref.~\cite{bourin2012unitary} and Eq.~(2.5) of Ref.~\cite{ando1994majorizations}):
\begin{lem}
	\label{lemma:decomposition}
	For every positive semi-definite matrix written in blocks we have the following decomposition
	\begin{equation}
	\left[
	\begin{array}{cc}
	A&X\\
	X^\dagger&B
	\end{array}
	\right]=
	U
	\left[
	\begin{array}{cc}
	A&0\\
	0&0
	\end{array}
	\right]U^\dagger+
	V\left[
	\begin{array}{cc}
	0&0\\
	0&B
	\end{array}
	\right]V^\dagger,
	\end{equation}
	for some unitary operators $U$ and $V$.
\end{lem}
\begin{lem}
	\label{lemma:KyFan}
	For $\v{\lambda}(A)$ denoting the vector of eigenvalues of $A$ arranged in a decreasing order we have
	\begin{equation}	
	\v{\lambda}(A)+\v{\lambda}(B)\succ \v{\lambda}(A+B).
	\end{equation}	
\end{lem}

We are now ready to present the proof of Theorem~\ref{thm:majo}.
\begin{proof}
	First, using Lemma~\ref{lemma:decomposition} iteratively one gets:
	\begin{equation}
	\renewcommand{\arraystretch}{1.3}
	\Scale[0.85]{
	\begin{array}{lll}
		\!\!\!\!d\cdot J_\Phi&=&U_1\left[\begin{array}{cccc}
			D^1 & 0 & \dots & 0\\
			0 & 0 & \dots & 0\\
			\vdots & \vdots & \ddots & \vdots\\
			0 & 0 & \dots & 0
		\end{array}\right]U_1^\dagger+
		U_2\left[\begin{array}{cccc}
			0 & 0 & \dots & 0\\
			0 & D^2 & \dots & 0\\
			\vdots & \vdots & \ddots & \vdots\\
			0 & 0 & \dots & 0
		\end{array}\right]U_2^\dagger\!\!\!\!\!\!\!\!\!\!\!\\
		&&+\dots+
		U_d\left[\begin{array}{cccc}
			0 & 0 & \dots & 0\\
			0 & 0 & \dots & 0\\
			\vdots & \vdots & \ddots & \vdots\\
			0 & 0 & \dots & D^d
		\end{array}\right]U_d^\dagger.
	\end{array}}
	\end{equation}
	Then, using unitary invariance of the spectrum, the fact that $\v{\lambda}_1\succ\v{\lambda'}_1$ and $\v{\lambda}_2\succ\v{\lambda'}_2$ induces $\v{\lambda}_1+\v{\lambda}_2\succ\v{\lambda'}_1+\v{\lambda'}_2$, and iteratively applying Lemma~\ref{lemma:KyFan} one arrives at Eq.~\eqref{eq:majo_bound}.
\end{proof}

\section{Coherifying qubit channels}
\label{app:qubits}

Before we proceed to deriving the results presented in the main text, let us first recall the following fact. Given two channels, $\Phi_1$ and $\Phi_2$, whose Jamio{\l}kowski states, $J_{\Phi_1}$ and $J_{\Phi_2}$, are connected via a local unitary acting on the second subsystem, \mbox{$J_{\Phi_2}=(\iden\otimes \bar{U})J_{\Phi_1}(\iden\otimes \bar{U})^\dagger$}, their Kraus decomposition satisfies
\begin{equation}
\label{eq:kraus_connection}
\Phi_1(\cdot)=\sum_i K_i(\cdot)K_i^\dagger,\quad \Phi_2(\cdot)=\sum_i K_i U^\dagger (\cdot) U K_i^\dagger.
\end{equation}

Now, using the block structure of the Jamio{\l}kowski state, Eq.~\eqref{eq:jamiol_blocks}, and taking into account the TP condition, \mbox{$\sum_i D^i=\iden$}, for a general qubit channel, we get:
\begin{equation}
J_\Phi=\frac{1}{2}\left[
\begin{array}{cc}
D^1&C^{12}\\
C^{21}&\iden-D^1
\end{array}\right], \quad
D^1=\left[
\begin{array}{cc}
a&c\\
c^*&\tilde{b}
\end{array}\right].
\end{equation}
Consider a unitary $U$ diagonalizing $D^1$, i.e., \mbox{$U D^1 U^\dagger=\diag{\lambda_1,\lambda_2}$}. Now, the same unitary will obviously also diagonalize $\iden-D^1$. Therefore, a $4\times 4$ unitary $V=\iden\otimes U$ diagonalizes the $D^i$ blocks of the Jamio{\l}kowski state $J_\Phi$:
\begin{equation}
V J_\Phi V^\dagger=\frac{1}{2}\left[
\begin{array}{cccc}
\lambda_1&0&&\\
0&\lambda_2&&\\
&&\tilde{\lambda_1}&0\\
&&0&\tilde{\lambda_2}
\end{array}
\right],
\end{equation}
where blank spaces mean arbitrary entries, \mbox{$0\leq \lambda_i\leq 1$} and \mbox{$\lambda_1+\lambda_2=a+\tilde{b}$}.

For $a\leq b$ we may obtain the optimally coherified state $J^\C_\Phi$ [with the spectrum saturating the bound given by $\v{\mu}^\succ(T)$ from Eq.~\eqref{eq:qubit_opt_spectrum}] in the following way. We choose $\lambda_1=a+\tilde{b}$ (resulting in $\lambda_2=0$), and set the off-diagonal element between $\lambda_1$ and $\tilde{\lambda}_2=1$ to the maximal value allowed by positivity, i.e., \mbox{$\sqrt{a+\tilde{b}}$}. As a result, $V J^\C_\Phi V^\dagger$ becomes a projector on two orthogonal pure states, which in turn means that the corresponding map is given by the decaying Kraus operators $L_1,L_2$ from Eq.~\eqref{eq:qubit_decaying}. Since $V=\iden\otimes U$, we can use Eq.~\eqref{eq:kraus_connection}, to find the Kraus of decomposition of $\Phi^\C$ given by $L_iU^\top$. Finally, $U$ is defined by \mbox{$UD^1U^\dagger=\mathrm{diag}(a+\tilde{b},0)$}, which is exactly the unitary given in Eq.~\eqref{eq:qubit_unitary} (note that, since $U$ is real, we have $U^\dagger=U^\top$)

Similarly, for $a\geq b$ we may choose $\lambda_1=1$ (resulting in $\lambda_2=a-b$), and set the off-diagonal element between $\lambda_1$ and $\tilde{\lambda}_2=\tilde{a}+b$ to the maximal value allowed by positivity, i.e., \mbox{$\sqrt{\tilde{a}+b}$}. As described in the main text, this then leads to the same results as in $a\leq b$ case, just with $a$ and $b$ exchanged, as well as with all $2\times 2$ matrices $X$ transformed by replacing $X_{kl}$ with $X_{\tilde{k}\tilde{l}}$.

\section{Coherifyng qutrit channels}
\label{app:qutrits}

\subsection{Cyclic matrices}

The general form of $D^i$ matrices is as follows:
\begin{equation*}
\Scale[0.92]{
D^1=
\left[
\begin{array}{ccc}
0 & 0 & 0\\
0 & a & x\\
0 & x^* & b
\end{array}
\right],~
D^2=
\left[
\begin{array}{ccc}
c & 0 & y\\
0 & 0 & 0\\
y^* & 0 & \tilde{b}
\end{array}
\right],~
D^3=\left[
\begin{array}{ccc}
\tilde{c} & z & 0\\
z^* & \tilde{a} & 0\\
0 & 0 & 0
\end{array}
\right].}
\end{equation*}
Clearly, in order to satisfy the TP condition, \mbox{$\sum_i D^i=\iden$}, we need \mbox{$x=y=z=0$}. Hence, the Jamio{\l}kowski state $J_\Phi$ can be recast in the following form (note that columns and rows number 1, 5 and 9, composed only of zeros, have been removed):
\begin{equation}
J_\Phi=\frac{1}{3}
\left[
\begin{array}{cccccc}
a&0&x_1&x_2&x_3&x_4\\
0&b&y_1&y_2&y_3&y_4\\
x_1^*&y_1^*&c&0&x_5&x_6\\
x_2^*&y_2^*&0&\tilde{b}&y_5&y_6\\
x_3^*&y_3^*&x_5^*&y_5^*&\tilde{c}&0\\
x_4^*&y_4^*&x_6^*&y_6^*&0&\tilde{a}
\end{array}
\right].
\end{equation}
Now, using Theorem~\ref{thm:majo}, we get:
\begin{equation}
\v{\mu}_{\mathrm{opt}}:=\frac{1}{3}[\mu,3-\mu,0,\dots,0]\succ\v{\lambda}(J_\Phi).
\end{equation}
where
\begin{equation}
\mu=\max(a,b)+\max(c,\tilde{b})+\max(\tilde{c},\tilde{b}).
\end{equation}
Moreover, one can construct optimally coherified matrix $J^\C_\Phi$ such that $\v{\lambda}(J^\C_\Phi)=\v{\mu}_{\mathrm{opt}}$. To do this one simply needs to group together the maximal/minimal terms of each \mbox{$2\times 2$} matrix and set the corresponding off-diagonal terms to the maximal values allowed by the positivity constraint. For example, if $a\geq b$, $\tilde{b}\geq c$ and $\tilde{c}\geq\tilde{a}$, one chooses 
\begin{eqnarray*}
x_2=\sqrt{a\tilde{b}},\quad&x_3=\sqrt{a\tilde{c}},\quad&y_5=\sqrt{\tilde{b}\tilde{c}},\\
y_1=\sqrt{bc},\quad&y_4=\sqrt{\tilde{a}b},\quad&x_6=\sqrt{\tilde{a}c},
\end{eqnarray*}
and sets the rest of off-diagonal terms to zero. Note that this is exactly the construction introduced in Sec.~\ref{sec:lower_bound} and illustrated in Fig.~\ref{fig:lower_bound}.

\subsection{Single-row matrices}

The general form of $D^i$ matrices is as follows:
\begin{equation*}
\Scale[0.92]{
D^1=
\left[
\begin{array}{ccc}
a & x & 0\\
x^* & b & 0\\
0 & 0 & 0
\end{array}
\right],~
D^2=
\left[
\begin{array}{ccc}
0 & 0 & 0\\
0 & 0 & 0\\
0 & 0 & c
\end{array}
\right],~
D^3=\left[
\begin{array}{ccc}
\tilde{a} & x' & y\\
x'^* & \tilde{b} & z\\
y^* & z^* & \tilde{c}
\end{array}
\right].}
\end{equation*}
Clearly, in order to satisfy the TP condition we need $x'=-x$ and $y=z=0$. We now note that a $3\times 3$ unitary matrix $U$ diagonalizing $D^1$,
\begin{equation}
\label{eq:diagonalising_U_single}
U=\left[
\begin{array}{ccc}
U_{11} & U_{12} & 0\\
U_{21} & U_{22} & 0\\
0 & 0 & 1
\end{array}
\right],\quad
UD^1U^\dagger=\left[
\begin{array}{ccc}
\lambda_1 & 0 & 0\\
0 & \lambda_2 & 0\\
0 & 0 & 0
\end{array}
\right],
\end{equation}
also diagonalizes $D^2$ (by keeping it unchanged) and $D^3$. Therefore, a $9\times 9$ unitary $V=\iden\otimes U$ diagonalizes the $D^i$ blocks of the Jamio{\l}kowski state $J_\Phi$:
\begin{equation}
\label{eq:jamiol_single_row}
V J_\Phi V^\dagger=\frac{1}{3}\left[
\begin{array}{ccccccccc}
\lambda_1&0&0&&&&&&\\
0&\lambda_2&0&&&&&&\\
0&0&0&&&&&&\\
&&&0&0&0&&&\\
&&&0&0&0&&&\\
&&&0&0&c&&&\\
&&&&&&\tilde{\lambda_1}&0&0\\
&&&&&&0&\tilde{\lambda_2}&0\\
&&&&&&0&0&\tilde{c}
\end{array}
\right],
\end{equation}
where blank spaces mean arbitrary entries and \mbox{$\lambda_1+\lambda_2=a+b$}. 

Without loss of generality let us assume that \mbox{$\lambda_1\geq\lambda_2$}. Then, using Theorem~\ref{thm:majo}, we get \mbox{$\v{\mu}(\lambda_1)\succ \v{\lambda}(J_\Phi)$} with $\v{\mu}^\downarrow(\lambda_1)$ given by
\begin{equation*}
\Scale[0.88]{
\frac{1}{3}[\lambda_1+\max(\tilde{\lambda_2},\tilde{c})+c,\lambda_2+\mathrm{med}(\tilde{\lambda_1},\tilde{\lambda_2},\tilde{c}),\min(\tilde{\lambda_1},\tilde{c}),0,\dots,0],}
\end{equation*}
where $\mathrm{med}$ denotes the second largest element of the set. Note also that, for fixed $a,b,c$, the vector $\v{\mu}(\lambda_1)$ is just a function of $\lambda_1$, since $\lambda_2=a+b-\lambda_1$. Now, maximising $\lambda_1$ maximizes both $\mu_1^\downarrow$ and $\mu_1^\downarrow+\mu_2^\downarrow$ (recall that $\mu_1^\downarrow+\mu_2^\downarrow+\mu_3^\downarrow$ is constant and equal to 1), and so $\v{\mu}(x)\succ\v{\mu}(y)$ for \mbox{$x\geq y$}. In order to find the optimal $\v{\mu}_{\mathrm{opt}}$ (optimal meaning that for all $\lambda_1$ we have $\v{\mu}_{\mathrm{opt}}\succ\v{\mu}(\lambda_1)$) we thus need to maximize $\lambda_1$. Recalling that we have two constraints, $0\leq\lambda_1\leq a+b$ and $\lambda_1\leq 1$, we arrive at two cases.

For $a+b\leq 1$ the maximal (and thus optimal) value of $\lambda_1$ is $a+b$, which also results in $\lambda_2=0$. The optimal bounding vector is then given by  
	\begin{equation}
	\Scale[0.83]{
	\!\!\!\v{\mu}_{\mathrm{opt}}=\frac{1}{3}[1+a+b+c,\max(\tilde{a}-b,\tilde{c}),\min(\tilde{a}-b,\tilde{c}),0,\dots,0].\!\!\!}
	\end{equation}
	Moreover, one can construct the Jamio{\l}kowski state $J^\C_\Phi$ that saturates this optimal bound, i.e., \mbox{$\v{\lambda}(J^\C_\Phi)=\v{\mu}_{\mathrm{opt}}$}. This can be achieved, again, by setting the adequate off-diagonal terms in Eq.~\eqref{eq:jamiol_single_row} to the maximal possible value allowed by the positivity condition. More precisely, we group $\lambda_1=a+b$, $c$ and $\tilde{\lambda_2}=1$ together, leaving the remaining two terms, $\tilde{\lambda_1}=\tilde{a}-b$ and $\tilde{c}$, ungrouped. As a result, $V J^\C_\Phi V^\dagger$ becomes a projector on three orthogonal pure states, which in turn means that the corresponding map is given by the following Kraus operators:
	\begin{equation*}
	K_1=\left[
	\begin{array}{ccc}
	\sqrt{a+b}&0&0\\
	0&0&\sqrt{c}\\
	0&1&0
	\end{array}
	\right],~
	K_2=\left[
	\begin{array}{ccc}
	0&0&0\\
	0&0&0\\
	\sqrt{\tilde{a}-b}&0&0
	\end{array}
	\right],
	\end{equation*}
	\begin{equation*}
	K_3=\left[
	\begin{array}{ccc}
	0&0&0\\
	0&0&0\\
	0&0&\sqrt{\tilde{c}}
	\end{array}
	\right].
	\end{equation*}
	Finally, using Eq.~\eqref{eq:kraus_connection} we conclude that the Kraus operators corresponding to the optimally coherified channel (with Jamio{\l}kowski state $J^\C_\Phi$) are given by $K_i U^\top$ with $K_i$ as above and $U$ defined by Eq.~\eqref{eq:diagonalising_U_single} with $\lambda_1=a+b$ and $\lambda_2=0$, i.e.,
	\begin{equation}
		U=\frac{1}{\sqrt{a+b}}\left[
		\begin{array}{ccc}
			\sqrt{a} & \sqrt{b} & 0\\
			\sqrt{b} & -\sqrt{a} & 0\\
			0 & 0 & \sqrt{a+b}
		\end{array}\right].
	\end{equation}

For $a+b\geq 1$ the maximal (and thus optimal) value of $\lambda_1$ is $1$, which also results in \mbox{$\lambda_2=a-\tilde{b}$}. The optimal bounding vector is then given by 
\begin{equation}
\Scale[0.88]{
\!\!\!\v{\mu}_{\mathrm{opt}}=\frac{1}{3}[\mu,3-\mu,0,\dots,0],\quad\mu=1+\max(\tilde{a}+\tilde{b},\tilde{c})+c.\!\!\!}
\end{equation}
The bound $\v{\mu}_{\mathrm{opt}}\succ \v{\lambda}(J_\Phi)$ can be saturated in a usual way -- by proper grouping of diagonal elements and setting the corresponding off-diagonal elements to the maximal value allowed by the positivity condition. If \mbox{$\tilde{c}\geq\tilde{a}+\tilde{b}$} then we group together $\lambda_1=1$, $c$ and $\tilde{c}$, with \mbox{$\lambda_2=a-\tilde{b}$} and \mbox{$\tilde{\lambda_2}=\tilde{a}+\tilde{b}$} forming the other group; otherwise we group together $\lambda_1$, $c$ and $\lambda_2$, with $\lambda_2$ and $\tilde{c}$ forming the other group. In the former case the resulting Kraus operators of the optimally coherified channel read
\begin{equation*}
K_1=\left[
\begin{array}{ccc}
1&0&0\\
0&0&\sqrt{c}\\
0&0&\sqrt{\tilde{c}}
\end{array}
\right]U^\top,~
K_2=\left[
\begin{array}{ccc}
0&\sqrt{a-\tilde{b}}&0\\
0&0&0\\
0&\sqrt{\tilde{a}+\tilde{b}}&0
\end{array}
\right]U^\top,
\end{equation*}
and in the latter case they read
\begin{equation*}
K_1=\left[
\begin{array}{ccc}
1&0&0\\
0&0&\sqrt{c}\\
0&\sqrt{\tilde{a}+\tilde{b}}&0
\end{array}
\right]U^\top,~
K_2=\left[
\begin{array}{ccc}
0&\sqrt{a-\tilde{b}}&0\\
0&0&0\\
0&0&\sqrt{\tilde{c}}
\end{array}
\right]U^\top,
\end{equation*}
with $U$ in both cases defined by Eq.~\eqref{eq:diagonalising_U_single} with $\lambda_1=1$ and \mbox{$\lambda_2=a-\tilde{b}$}, i.e.,
\begin{equation}
U=\frac{1}{\sqrt{\tilde{a}+\tilde{b}}}\left[
\begin{array}{ccc}
\sqrt{\tilde{b}} & \sqrt{\tilde{a}} & 0\\
\sqrt{\tilde{a}} & -\sqrt{\tilde{b}} & 0\\
0 & 0 & \sqrt{\tilde{a}+\tilde{b}}
\end{array}\right].
\end{equation}

\subsection{Double-row matrices}

The general form of $D^i$ matrices is as follows:	
\begin{equation*}
\Scale[0.92]{
D^1=
\left[
\begin{array}{ccc}
a & x & y\\
x^* & b & z\\
y^* & z^* & c
\end{array}
\right],~
D^2=
\left[
\begin{array}{ccc}
\tilde{a} & -x & -y\\
-x^* & \tilde{b} & -z\\
-y^* & -z^* & \tilde{c}
\end{array}
\right],~
D^3=0.}
\end{equation*}
We now note that a $3\times 3$ unitary matrix $U$ diagonalizing~$D^1$,
\begin{equation}
\label{eq:diagonalising_U_double}
\Scale[0.92]{
U=\left[
\begin{array}{ccc}
U_{11} & U_{12} & U_{13}\\
U_{21} & U_{22} & U_{23}\\
U_{31} & U_{32} & U_{33}
\end{array}
\right],\quad
UD^1U^\dagger=\left[
\begin{array}{ccc}
\lambda_1 & 0 & 0\\
0 & \lambda_2 & 0\\
0 & 0 & \lambda_3
\end{array}
\right],}
\end{equation}
also diagonalizes $D^2$ and $D^3$ (by keeping it unchanged). Therefore, a $9\times 9$ unitary $V=\iden\otimes U$ diagonalizes the $D^i$ blocks of the Jamio{\l}kowski state $J_\Phi$:
\begin{equation}
\label{eq:jamiol_double_row}
V J_\Phi V^\dagger=\frac{1}{3}\left[
\begin{array}{ccccccccc}
\lambda_1&0&0&&&&&&\\
0&\lambda_2&0&&&&&&\\
0&0&\lambda_3&&&&&&\\
&&&\tilde{\lambda_1}&0&0&&&\\
&&&0&\tilde{\lambda_2}&0&&&\\
&&&0&0&\tilde{\lambda_3}&&&\\
&&&&&&0&0&0\\
&&&&&&0&0&0\\
&&&&&&0&0&0
\end{array}
\right],
\end{equation}
where blank spaces mean arbitrary entries and \mbox{$\lambda_1+\lambda_2+\lambda_3=a+b+c$}. To shorten the notation we define \mbox{$s:=a+b+c$}.

Without loss of generality we may assume \mbox{$\lambda_1\geq\lambda_2\geq\lambda_3$}, so that \mbox{$\tilde{\lambda_1}\leq\tilde{\lambda_2}\leq\tilde{\lambda_3}$}. Then, using Theorem~\ref{thm:majo}, we have
\begin{equation}
\Scale[0.92]{
\v{\mu}(\lambda_1,\lambda_3)=\frac{1}{3}[\lambda_1+\tilde{\lambda_3},1,\lambda_3+\tilde{\lambda_1},0,\dots,0]\succ\v{\lambda}(J_\Phi).}
\end{equation}
Now, we observe that $\v{\mu}(x,y)\succ\v{\mu}(x',y')$ for $x\geq x'$ and $y\leq y'$. We thus aim at maximising the largest eigenvalue of $D^1$ while minimizing its smallest eigenvalue. Again, noting that we are constrained by $0\leq \lambda_i\leq s$ and $\lambda_i\leq 1$ we arrive at three distinct cases dependent on the value of $s$.

For $s\leq 1$ the maximal (and thus optimal) value of $\lambda_1$ is $s$, which also results in $\lambda_2=\lambda_3=0$. The optimal bounding vector is then given by  
\begin{equation}
\v{\mu}_{\mathrm{opt}}=\frac{1}{3}[1+s,1,1-s,0,\dots,0].
\end{equation}
The above optimal spectrum can be realized by the Jamio{\l}kowski state $J^\C_\Phi$ by simply setting in Eq.~\eqref{eq:jamiol_double_row} the off-diagonal terms between \mbox{$\lambda_1=s$} and $\tilde{\lambda_2}=1$ (or $\tilde{\lambda_3}=1$) to \mbox{$\sqrt{s}$}. Recalling the relation between the Kraus operators corresponding to Jamio{\l}kowski states connected via a local unitary, Eq.~\eqref{eq:kraus_connection}, we find that the Kraus decomposition of the optimally coherified channel is given by:
\begin{equation*}
K_1=\left[
\begin{array}{ccc}
\sqrt{s}&0&0\\
0&1&0\\
0&0&0
\end{array}
\right]U^\top,~
K_2=\left[
\begin{array}{ccc}
0&0&0\\
0&0&1\\
0&0&0
\end{array}
\right]U^\top,
\end{equation*}
\begin{equation*}
K_3=\left[
\begin{array}{ccc}
0&0&0\\
\sqrt{1-s}&0&0\\
0&0&0
\end{array}
\right]U^\top.
\end{equation*}
with $U$ being a unitary such that
	\begin{equation}
	U=\frac{1}{\sqrt{s}}\left[\begin{array}{ccc}
	\sqrt{a}&\times&\times\\
	\sqrt{b}&\times&\times\\
	\sqrt{c}&\times&\times
	\end{array}\right],
	\end{equation}
and $\times$ denoting arbitrary entries as long as $U$ stays unitary, e.g.,

\begin{equation*}
U=\left[\begin{array}{ccc}
\sqrt{\frac{a}{s}}&-\sqrt{\frac{b}{a+b}}&-\sqrt{\frac{ac}{(a+b)s}}\\
\sqrt{\frac{b}{s}}&\sqrt{\frac{a}{a+b}}&-\sqrt{\frac{bc}{(a+b)s}}\\
\sqrt{\frac{c}{s}}&0&\frac{a+b}{\sqrt{(a+b)s}}\\
\end{array}\right].
\end{equation*}
Also note that the position of $1$ in matrices describing $K_1$ and $K_2$ can be exchanged.
	
For $s\geq 2$ the optimal values are $\lambda_1=\lambda_2=1$ and $\lambda_3=s-2$. The optimal bounding vector is then given by 
\begin{equation}
\v{\mu}_{\mathrm{opt}}=\frac{1}{3}[4-s,1,s-2,0,\dots,0],
\end{equation}
which can be achieved by the coherified Jamio{\l}kowski state in an analogous way to the first case. This leads to the following decomposition of $\Phi^\C$ into the set of Kraus operators, 
\begin{equation*}
K_1=\left[
\begin{array}{ccc}
1&0&0\\
0&0&\sqrt{3-s}\\
0&0&0
\end{array}
\right]U^\top,~
K_2=\left[
\begin{array}{ccc}
0&1&0\\
0&0&0\\
0&0&0
\end{array}
\right]U^\top,
\end{equation*}
\begin{equation*}
K_3=\left[
\begin{array}{ccc}
0&0&\sqrt{s-2}\\
0&0&0\\
0&0&0
\end{array}
\right]U^\top,
\end{equation*}
with $U$ being a unitary such that
	\begin{equation}
	U=\frac{1}{\sqrt{3-s}}\left[\begin{array}{ccc}
	\times&\times&\sqrt{\tilde{a}}\\
	\times&\times&\sqrt{\tilde{b}}\\
	\times&\times&\sqrt{\tilde{c}}
	\end{array}\right].
	\end{equation}
and $\times$ denoting arbitrary entries as long as $U$ stays unitary, e.g.,
	\begin{equation}
	U=\left[\begin{array}{ccc}
	-\sqrt{\frac{\tilde{a}\tilde{c}}{(\tilde{a}+\tilde{b})(3-s)}}&-\sqrt{\frac{\tilde{b}}{\tilde{a}+\tilde{b}}}&\sqrt{\frac{\tilde{a}}{3-s}}\\
	-\sqrt{\frac{\tilde{b}\tilde{c}}{(\tilde{a}+\tilde{b})(3-s)}}&\sqrt{\frac{\tilde{a}}{\tilde{a}+\tilde{b}}}&\sqrt{\frac{\tilde{b}}{3-s}}\\
	\frac{\tilde{a}+\tilde{b}}{\sqrt{(\tilde{a}+\tilde{b})(3-s)}}&0&\sqrt{\frac{\tilde{c}}{3-s}}
	\end{array}\right].
	\end{equation}
Again, we note that the position of $1$ in matrices describing $K_1$ and $K_2$ can be exchanged.
	
Finally, for $1< s< 2$ the optimal values are $\lambda_1=1$, \mbox{$\lambda_2=s-1$} and $\lambda_3=0$. The optimal bounding vector is then given by 
\begin{equation}
\v{\mu}_{\mathrm{opt}}=\frac{1}{3}[2,1,0,\dots,0].
\end{equation}
The above spectrum can be achieved by the optimally coherified Jamio{\l}kowski state $J^\C_\Phi$ by setting the off-diagonal terms in Eq.~\eqref{eq:jamiol_double_row} appropriately. More precisely, one chooses the term between \mbox{$\lambda_1=1$} and $\tilde{\lambda_3}=1$ to be \mbox{$1$}, and the term between \mbox{$\lambda_2=s-1$} and $\tilde{\lambda_2}=2-s$ to be \mbox{$\sqrt{(s-1)(2-s)}$}. The resulting Kraus operators are given by
\begin{equation}
K_1=\left[
\begin{array}{ccc}
1&0&0\\
0&0&1\\
0&0&0
\end{array}
\right]U^\top,~
K_2=\left[
\begin{array}{ccc}
0&\sqrt{s-1}&0\\
0&\sqrt{2-s}&0\\
0&0&0
\end{array}
\right]U^\top,
\end{equation}
with
\begin{equation}
U=\left[
\begin{array}{ccc}
\sqrt{\frac{\tilde{b}}{\tilde{a}+\tilde{b}}}&\sqrt{\frac{\tilde{a}(a-\tilde{b})}{(\tilde{a}+\tilde{b})(s-1)}}&\sqrt{\frac{\tilde{a}c}{(\tilde{a}+\tilde{b})(s-1)}}\\
-\sqrt{\frac{\tilde{a}}{\tilde{a}+\tilde{b}}}&\sqrt{\frac{\tilde{b}(a-\tilde{b})}{(\tilde{a}+\tilde{b})(s-1)}}&\sqrt{\frac{\tilde{b}c}{(\tilde{a}+\tilde{b})(s-1)}}\\
0&\sqrt{\frac{c}{s-1}}&-\sqrt{\frac{a-\tilde{b}}{s-1}}
\end{array}
\right]
\end{equation} 
if $a+b\geq 1$ and
\begin{equation}
U=\left[
\begin{array}{ccc}
\sqrt{\frac{a\tilde{c}}{(a+b)(2-s)}}&\sqrt{\frac{a(\tilde{a}-b)}{(a+b)(2-s)}}&\sqrt{\frac{b}{a+b}}\\
\sqrt{\frac{b\tilde{c}}{(a+b)(2-s)}}&\sqrt{\frac{b(\tilde{a}-b)}{(a+b)(2-s)}}&-\sqrt{\frac{a}{a+b}}\\
-\sqrt{\frac{\tilde{a}-b}{2-s}}&\sqrt{\frac{\tilde{c}}{2-s}}&0
\end{array}
\right]
\end{equation}
if $a+b\leq 1$.

\section{Polygon constraints}
\label{app:polygon}

First, we derive the expression for the purity bound, Eq.~\eqref{eq:purity_bound}. The expression for $\Delta_1$, Eq.~\eqref{eq:purity_deficit1}, comes directly from the fact that \mbox{$|D^i_{kl}|^2\leq\alpha^{i}_{kl}T_{ik}T_{il}$}. To obtain $\Delta_2$, Eq.~\eqref{eq:purity_deficit2}, let us start by parametrizing the matrix from Eq.~\eqref{eq:Sylvesters}, i.e., the $3\times 3$ submatrix of $J_\Phi$, in the following way
\begin{equation}
A:=\frac{1}{d}\left[
\begin{array}{ccc}
T_{ik}&x\sqrt{T_{ik}T_{il}}&y\sqrt{T_{ik}T_{jm}}\\
x\sqrt{T_{ik}T_{il}}&T_{il}&z\sqrt{T_{il}T_{jm}}\\
y\sqrt{T_{ik}T_{jm}}&z\sqrt{T_{il}T_{jm}}&T_{jm}
\end{array}
\right],
\end{equation}
with \mbox{$0\leq x,y,z\leq 1$}. We then have $\det A\geq 0$ if and only if
\begin{equation}
\label{eq:positive_det}
1 + 2 x y z - x^2 - y^2 - z^2\geq 0.
\end{equation}
Now, our aim is to upper-bound the squared moduli of the off-diagonal terms of $A$ (for fixed $T$), given the above constraint and the fact that $x\leq\alpha^i_{kl}$ for some $\alpha^i_{kl}<1$. First, assume that $x$ is fixed, so that effectively we want to find the maximum of \mbox{$T_{ik}y^2+T_{il}z^2$} (in fact, the optimal choice is to maximize $x$, i.e., set $x=\alpha^i_{kl}$). It is straightforward to check that it is achieved at the boundary of the constraint, i.e., when Eq.~\eqref{eq:positive_det} becomes an equality. One can then solve for $y$, substitute it to \mbox{$T_{ik}y^2+T_{il}z^2$}, and find the maximum of the resulting expression over~$z$. This leads to the following bound on the off-diagonal terms of~$A$:
\begin{equation}
|A_{13}|^2+|A_{23}|^2\leq\frac{T_{jm}}{2d^2}\left(T_{ik}+T_{il}+\beta^i_{kl}\right),
\end{equation}
with $\beta^i_{kl}$ defined in Eq.~\eqref{eq:beta}. As in order to achieve unit purity one needs \mbox{$|A_{13}|^2+|A_{23}|^2=T_{jm}(T_{ik}+T_{il})/d^2$}, the above bound leads to the following deficit of purity:
\begin{equation}
\delta^i_{kl}:=\frac{T_{jm}}{2d^2}\left(T_{ik}+T_{il}-\beta^i_{kl}\right).
\end{equation}
Finally, the above deficit adds up for every choice of $T_{jm}$ not equal to $T_{ik}$ or $T_{il}$ (i.e., for all off-diagonal elements sharing row or column with $T_{ik}$ or $T_{il}$ in Fig.~\ref{fig:bistochastic_constraint}b), so that using \mbox{$\sum_{i,j}T_{ij}=d$}, we finally arrive at Eq.~\eqref{eq:purity_deficit2}.

We now proceed to the proof of the majorization bound, Eq.~\eqref{eq:majo_bistoch}. Note that, using Lemma~\ref{lemma:decomposition} from Appendix~\ref{app:majo_proof}, we can rewrite $D^i$ (up to permutations) as
\begin{equation}
D^i=
\left[
\begin{array}{cc}
A&X\\
X^\dagger&B
\end{array}
\right]
=
U
\left[
\begin{array}{cc}
A&0\\
0&0
\end{array}
\right]U^\dagger+
V\left[
\begin{array}{cc}
0&0\\
0&B
\end{array}
\right]V^\dagger,
\end{equation}
with 
\begin{equation}
A=\left[
\begin{array}{cc}
T_{ik} & \sqrt{\alpha^i_{kl}T_{ik}T_{il}}\\
\sqrt{\alpha^i_{kl}T_{ik}T_{il}} & T_{il}
\end{array}
\right],
\end{equation}
and $U,V$ being unitary. The eigenvalues of $A$ are given by
\begin{equation}
\v{\lambda}(A)=\frac{1}{2}[T_{ik}+T_{il}+\beta^i_{kl},T_{ik}+T_{il}-\beta^i_{kl},0,\dots,0],
\end{equation}
whereas the largest eigenvalue of $B$ is constrained by
\begin{equation}
\Scale[0.9]{
\lambda_1(B)\leq\tr{B}=\tr{D^i}-T_{ik}-T_{il}=1-T_{ik}-T_{il},}
\end{equation} \normalsize
where we used the fact that $T$ is bistochastic. Thus, using Lemma~\ref{lemma:KyFan} and choosing $k=k_i$ and $l=l_i$ minimizing $\alpha^i_{kl}$, we arrive at Eq.~\eqref{eq:majo_bistoch}.

Note that the above construction can be easily generalized to cases where for a given $D^i$ there are many pairs $(k,l)$ for which $\alpha^i_{kl}<1$. Instead of a $2\times 2$ matrix $A$, one simply chooses it to contain all off-diagonal elements of $D^i$ that are constrained beyond the positivity condition, finds its eigenvalues, and obtains a tighter bound using Lemma~\ref{lemma:KyFan} again.

\bibliography{Bibliography_channels}

\end{document}